\documentclass[11pt,letterpaper]{article}

\usepackage{booktabs}
\usepackage{enumerate}
\usepackage{authblk}
\usepackage{amsmath,amsthm,amsfonts,amssymb,color,hyperref,anysize,enumerate,graphicx,epstopdf,algorithm,algpseudocode,cleveref,multirow}
\usepackage[usenames,dvipsnames]{xcolor}
\usepackage[margin=1in]{geometry}

\newcommand{\rrplusn}{\mathbb{R}_+^n}
\newcommand{\bps}{\mathbf{p}^*}
\newcommand{\ps}{p^*}

\newcommand{\rhoi}{{\rho_i}}

\newcommand{\rint}{\mathsf{rint}}
\newcommand{\inner}[2]{\left\langle ~#1 ~,~ #2~\right\rangle}
\newcommand{\KL}{\mathrm{KL}}

\newcommand{\hide}[1]{}

\newcommand{\rr}{\mathbb{R}}

\newcommand{\la}{\leftarrow}
\newcommand{\ra}{\rightarrow}

\renewcommand{\bold}[1]{\textbf{#1}}

\newcommand{\PPAD}{\textsf{PPAD}}

\DeclareMathOperator*{\argmin}{arg\,min}
\DeclareMathOperator*{\argmax}{arg\,max}

\newtheorem{defn}{Definition}

\newcommand{\bbb}{\mathbf{b}}

\newcommand{\bbe}{\mathbf{e}}

\newcommand{\bbp}{\mathbf{p}}

\newcommand{\bbx}{\mathbf{x}}
\newcommand{\bby}{\mathbf{y}}

\newcommand{\aij}{a_{ij}}
\newcommand{\aik}{a_{ik}}
\newcommand{\bij}{b_{ij}}

\newcommand{\bijt}{b_{ij}^t}
\newcommand{\bikt}{b_{ik}^t}
\newcommand{\cij}{c_{ij}}
\newcommand{\dg}{d_g}

\newcommand{\ei}{e_{i}}
\newcommand{\pj}{p_{j}}
\newcommand{\pjt}{p_{j}^t}
\newcommand{\pkt}{p_{k}^t}
\newcommand{\xij}{x_{ij}}

\newcommand{\xijt}{x_{ij}^t}
\newcommand{\xikt}{x_{ik}^t}

\newcommand{\zj}{z_j}

\newcommand{\hsp}{\hspace*{0.1in}}
\newcommand{\hhsp}{\hspace*{0.05in}}

\newcommand{\innerprod}[2]{\langle #1,#2\rangle}

\newcommand\numberthis{\addtocounter{equation}{1}\tag{\theequation}}
\newtheorem{theorem}{Theorem}[section]
\newtheorem{lemma}{Lemma}[section]

\newtheorem{corollary}{Corollary}[section]

\newcommand*\samethanks[1][\value{footnote}]{\footnotemark[#1]}

\begin{document}
\title{Dynamics of Distributed Updating in Fisher Markets}  
\author[1]{Yun Kuen Cheung}
\author[2]{Richard Cole\thanks{The work of Richard Cole and Yixin Tao was supported in part by NSF grant CCF-1527568.}}
\author[2]{Yixin Tao\samethanks}
\affil[1]{Max Planck Institute for Informatics, Saarland Informatics Campus}
\affil[2]{Courant Institute, NYU}
\maketitle

\begin{abstract}
A major goal in Algorithmic Game Theory is to justify equilibrium concepts from an algorithmic and complexity perspective.
One appealing approach is to identify natural distributed algorithms that converge quickly to an equilibrium.
This paper established new convergence results for two generalizations of Proportional Response in Fisher markets with buyers having CES utility functions.
The starting points are respectively a new convex and a new convex-concave formulation of such markets.
The two generalizations correspond to suitable mirror descent algorithms applied to these formulations.
Several of our new results are a consequence of new notions of strong Bregman convexity
and of strong Bregman convex-concave functions, and associated linear rates of convergence, which may be of independent interest.

Among other results, we analyze a damped generalized  Proportional Response and show a linear rate of convergence in
a Fisher market with buyers whose utility functions cover the full spectrum of CES utilities aside the extremes of linear and Leontief utilities;
when these utilities are included, we obtain an empirical $O(1/T)$ rate of convergence.

\smallskip

\noindent\bold{Keywords.} Proportional Response, mirror descent, Bregman divergence, Fisher market
\end{abstract}

\section{Introduction}\label{sec:intro}

One of the most important results in Algorithmic Game Theory is the \PPAD-hardness of finding a Nash Equilibrium
in finite games~\cite{DaskalakisGP09,ChenDT2009},
which serves as a strong evidence that there is no efficient algorithm to compute Nash Equilibria.
Similar hardness results have been established for markets~\cite{CSVY2006,CDDT2009,CT2009,VY2011,CPY2017}.
By viewing the players and the environment collectively as implicitly performing a computation, 
these hardness results indicate that, in general, a market cannot reach an equilibrium quickly. 
In Kamal Jain's words: ``If your laptop cannot find it, neither can the market''~\cite[Chapter 2.1]{NRTV07}.

As a result, a lot of attention has been given to the design of polynomial-time algorithms
to find equilibria, either exactly or approximately, for specific families of games and markets.
Most of these algorithms can be categorized as either simplex-like (e.g., Lemke-Howson), numerical methods
(e.g., the interior-point method or the ellipsoid method), 
or some carefully-crafted combinatorial algorithm (e.g., flow-based algorithms
for computing a market equilibrium for linear utility functions).

However, it seems implausible that these algorithms describe the implicit computations in games or markets.
For many markets would appear to have a highly distributed environment, or need to make rapid decisions on an ongoing basis.
These features would appear to preclude computations which require centralized coordination,
which is essential for the three categories of algorithms above.
Consequently, in order to justify equilibrium concepts, we want natural algorithms
which could plausibly be running (in an implicit form) in the associated distributed environments.

This paper will focus on Fisher markets 
(or economies\footnote{In the CS literature the term market has been widely used to refer to economies;
we follow this practice.}).
In a Fisher market there are buyers who start with money which they have no desire to keep,
and sellers who have goods to sell, which they wish to sell in their entirety for money.
This is a modest generalization of the notion of 
Competitive Equilibrium from Equal Incomes (CEEI)~\cite{ZH79,Varian74}.
In prior work on computing equilibria for these markets,
there has been a particular focus on Eisenberg-Gale markets, a term
coined by Jain and Vazirani, and their generalizations~\cite{JainVaz07};
Eisenberg-Gale markets are Fisher markets in which demands are determined by homothetic utility functions.
The latter markets have been seen to capture the notion of proportional fairness,
as defined in the networking community~\cite{Kelly-PF},
which is also equivalent to the optimum Nash Social Welfare~\cite{Nash50,KN79}.

Two natural dynamics have been studied in the context of Fisher markets.
The first, which is perhaps the most intuitive candidate for a natural algorithm, is tatonnement,
in which the price of a good is raised if the demand exceeds the supply of the good, and decreased if the demand is too small.
Implicitly, buyers' demands are assumed to be a best-response to the current prices.
This highly intuitive algorithm was proposed by Walras well over a century ago~\cite{Walras1874}.

Proportional Response, in contrast, is a buyer-oriented update,
originally analyzed in an effort to explain the behavior of peer-to-peer networks~\cite{WZ2007,LLSB08}.
Here, buyers update their spending in proportion to the contribution each good makes to its current utility.
While its meaning is clear for linear and other separable utilities,
for other classes of utilities this needs more interpretation, which we provide
in Section~\ref{sec:prelim}. Here prices are assumed to equal the current spending.
An $O(1/T)$ rate of convergence was shown in~\cite{BDX2011} for Fisher markets with buyers having
linear utilities, and  for the substitutes domain excluding linear utilities,
a faster linear rate (i.e., $\exp(-\Omega(T))$ rate) of convergence was shown in~\cite{Zhang2011}.

This paper continues the exploration of the connection between distributed dynamic processes and
convex optimization, and more specifically the relation of Proportional Response to mirror descent.

Our first set of results starts by rederiving Zhang's bounds for CES substitutes utilities,
by showing that for this setting Proportional Response amounts to mirror descent on a suitable convex function.
To achieve the linear rate of convergence he obtained, we need to go beyond the standard $O(1/T)$ rate of convergence
for mirror descent with a Bregman divergence. 
We proceed by analogy with gradient descent.
Gradient descent with a Lipschitz constraint on the gradients
guarantees only an $O(1/T)$ rate of convergence, but
a faster linear rate of convergence is obtained
when the objective function $f$ is strongly convex.
For mirror descent with Bregman Divergences we introduce the notion of 
\emph{strong Bregman convexity} and show that it also leads to a linear convergence rate.
It turns out that the convex function associated with the CES substitutes utilities
satisfies strong Bregman convexity, thereby obtaining Zhang's bound anew.
In addition, for complementary CES utilities, the same now concave function satisfies an analogous
strong Bregman concavity property, which also yields a linear rate of convergence for these utilities.
In addition, if we include linear utilities in the substitutes utilities, we obtain an $O(1/T)$ rate of convergence;
likewise, including Leontief utilities in the complementary utilities also yields an $O(1/T)$ rate.

Next, we seek to handle substitute and complementary CES utilities simultaneously.
The challenge we face is that the objective function used for the first set of results becomes a mixed
concave-convex function in this setting, and the equilibrium corresponds to a saddle point of this function.
We introduce the further notion of strongly-Bregman convex-concave functions,
and for these functions we obtain a linear rate of convergence to the saddle point.
Again, our objective function of the mixed CES utilities satisfies this property, 
thereby yielding linear convergence, albeit now for a \emph{Damped} Proportional Response,
rather than the undamped Proportional Response analyzed in the first set of results.
Here, including linear utilities and Leontief utilities yields an \emph{empirical} $O(1/T)$ rate of convergence.

We note that our results are not a straightforward application of the existing mirror descent toolbox.
The Bregman notions and the related convergence results in this paper are new.
While the results for strong Bregman convex (resp.~concave) functions are natural generalizations of gradient descent (resp.~ascent)
on standard strong convex (resp.~concave) functions,
the technique for demonstrating convergence for strong Bregman convex-concave functions appears to be new.
It is not evident that suitable damping (i.e., reducing the step-size) permits a clean convergence analysis.
Indeed, results showing linear point-wise convergence on convex-concave functions are rare; in fact,
the only such work we are aware of is~\cite{gidel2017frank}.
We believe the new notions and convergence results for optimization problems may be of wider interest.

\paragraph{Roadmap}
In Section~\ref{sec:prelim} we give necessary definitions and notation,
and in Section~\ref{sec:results} we state our results precisely.
Then, in Section~\ref{sec:related-work}, we describe related work.
In Section~\ref{sec:strong-Breg}  we carry out the analysis of mirror descent
when Strong Bregman convexity holds, and then derive the convergence behavior of
Proportional Response in each of the substitutes and complements domains.
In Section~\ref{sec:saddle-point}, we perform an analogous analysis
for strong Bregman convex-concave functions, and deduce the convergence behavior
 of a Damped Proportional Response in combined substitutes and complements domains.
Finally, in Section~\ref{sec:discussion} we 
discuss the rates of convergence under some alternate measures.
\section{Notation and Definitions}\label{sec:prelim}

We use bold symbols, e.g., $\bbp,\bbx,\bbe$, to denote vectors.

\paragraph{Fisher Market}
In a Fisher market, there are $n$ perfectly divisible goods and $m$ buyers.
Without loss of generality, the supply of each good $j$ is normalized to be one unit.
Each buyer $i$ has a utility function $u_i:\rrplusn \ra \rr$, and a budget of size $e_i$.
At any given price vector $\bbp\in \rrplusn$, each buyer purchases a maximum utility
affordable collection of goods.
More precisely, $\bbx_i\in \rrplusn$ is said to be a demand of buyer $i$ if $~\bbx_i~\in~\argmax_{\bbx':~\bbx'\cdot \bbp \le e_i}~u_i(\bbx')$.

A price vector $\bps \in \rrplusn$ is called a \emph{market equilibrium} if at $\bps$, there exists a demand $\bbx_i$ of each buyer $i$ such that
$$
\ps_j > 0~~\Rightarrow~~\sum_{i=1}^m x_{ij} ~=~ 1
\hsp\hsp\text{and}\hsp\hsp
\ps_j = 0~~\Rightarrow~~\sum_{i=1}^m x_{ij} ~\le~ 1.
$$
The collection of $\bbx_i$ is said to be an \emph{equilibrium allocation} to the buyers.

\paragraph{CES utilities}
In this paper, each buyer $i$'s utility function is of the form
$$
u_i(\bbx_i) ~=~ \left(\sum_{j=1}^n a_{ij} \cdot (x_{ij})^\rhoi\right)^{1/\rhoi},
$$
for some $-\infty \le \rhoi \le 1$.
$u_i(\bbx_i)$ is called a Constant Elasticity of Substitution (CES) utility function.
They are a class of utility functions often used in economic analysis.
The limit as $\rhoi \ra -\infty$ is called a Leontief utility, usually written as
$u_i(\bbx_i) = \min_j \frac{\xij}{c_{ij}}$ \footnote{Here, the utility function $u_i(\bbx) = \min_j \frac{\xij}{c_{ij}}$ can be seen as the limit of $u_i(\bbx) = \Big( \sum_j \left(\frac{\xij}{c_{ij}}\right)^{\rho_i}\Big)^{\frac{1}{\rho_i}}$ as $\rho_i$ tends to $-\infty$.}; and the limit as $\rhoi \ra 0$ is called a Cobb-Douglas utility, 
usually written as $\prod_j {\xij}^{a_{ij}}$, with $\sum_j \aij = 1$. 
The utilities with $\rhoi\ge 0$ capture goods that are substitutes, 
and those with $\rhoi \le 0$ goods that are
complements.

\paragraph{Notation} Buyer $i$'s spending on good $j$, denoted by $\bij$, is given by 
$\bij = \xij \cdot \pj$.
Also, $\zj = \sum_i \xij -1$ denotes the excess demand for good $j$.
We sometimes index prices, spending, and demands by $t$ to indicate the relevant value at time $t$.
Finally, we use a superscript of $^*$ to indicate an equilibrium value.

\paragraph{Bregman Divergence and Mirror Descent}
Let $C$ be a compact and convex set. Given a differentiable convex function $h(\bbx)$ with domain $C$,
the \emph{Bregman divergence} generated by kernel $h$ is denoted by $d_h$, 
and is defined as:
$$
d_h(\bbx,\bby) ~=~ h(\bbx) ~-~ \left[~h(\bby) ~+~ \inner{\nabla h(\bby)}{\bbx-\bby}~\right],~~~~\forall \bbx\in C\hhsp\text{and}\hhsp\bby \in \rint(C),
$$
where $\rint(C)$ is the relative interior of $C$.
We note that, in general, $d_h$ is asymmetric, i.e., possibly $d_h(\bbx,\bby) \neq d_h(\bby,\bbx)$.
In this paper, we use the Kullback-Leibler or KL divergence extensively;
it is the Bregman divergence generated by $h(\bbx) = \sum_j (x_j \cdot \ln x_j - x_j)$.
When $\sum_j x_j ~=~ \sum_j y_j$, the explicit formula is:
$$
\KL(\bbx || \bby) ~:=~ \sum_j x_j \cdot \ln \frac{x_j}{y_j}.
$$

For the problem of minimizing a convex function $f(\bbx)$ subject to $\bbx\in C$, the mirror descent method w.r.t.~Bregman divergence $d_h$
is given by the following update rule:
\begin{equation}\label{eq:ggd-update-rule}
\bbx^{t+1} ~=~ \argmin_{\bbx\in C} ~\left\{~f(\bbx^t) ~+~ \inner{\nabla f(\bbx^t)}{\bbx-\bbx^t} ~+~ \frac{1}{\Gamma_t} \cdot d_h(\bbx,\bbx^t)~\right\},
\end{equation}
where $\Gamma_t > 0$, and may be dependent on $t$.

\paragraph{Proportional Response}
For linear utility functions, Proportional Response
is the dynamic given by the spending update rule:
\[
b_{ij}^{t+1} = \ei \cdot \frac {\aij \xijt} {\sum_k \aik \xikt}
= \ei \cdot \frac {\aij \frac {\bijt}{\pjt}} {\sum_k \aik \frac {\bikt}{\pkt}}
\hsp \text{with } p_k^{t} = \sum_i b_{ik}^{t}.
\]
For substitutes CES utilities,~\cite{Zhang2011} generalized this rule to:
\begin{equation}
\label{eqn:subst-CES-PR-rule}
b_{ij}^{t+1} =  \ei\cdot  \frac {\aij (\xijt)^\rhoi} {\sum_k \aik (\xikt)^\rhoi}
= \ei\cdot  \frac {\aij \left(\frac {\bijt}{\pjt}\right)^\rhoi} 
              {\sum_k \aik \left(\frac {\bikt}{\pkt}\right)^\rhoi}
\end{equation}
obtaining a linear convergence rate for the resulting dynamic, assuming $0 < \rhoi < 1$.
The above rule has a natural distributed interpretation in the Fisher market setting:
in each round, each buyer splits her spending on different goods in proportion to the values of $\aik (\xikt)^\rhoi$.
The seller of good $j$ then allocates the good to buyers in proportion to the spending received from each buyer.
\section{Results}\label{sec:results}

\subsection{Proportional Response}

It is natural to seek to extend the Proportional Response rule~\eqref{eqn:subst-CES-PR-rule}
to the complementary domain, but this rule does not lead to convergent behavior in general.
Set $\rho = -1$.
Suppose there are two buyers and two items. 
Both buyers have the same preference for each item and the same budgets;
i.e.\ $a_{11}=a_{12}=a_{21}=a_{22}=\frac 12$, and $e_1 = e_2 = 1$.
Initially, at time $t=0$, suppose that $b_{11}^{(0)}= \frac 14$, $b_{12}^{(0)}= \frac 34$,
$b_{21}^{(0)}= \frac 34$, and $b_{22}^{(0)}= \frac 14$.
A simple calculation shows that applying update rule \eqref{eqn:subst-CES-PR-rule}
gives $b_{11}^{(1)}= \frac 34$, $b_{12}^{(1)}= \frac 14$,
$b_{21}^{(1)}= \frac 14$, and $b_{22}^{(1)}= \frac 34$.
So this simple example shows that in this setting, the spending will not converge to the market equilibrium;
rather, it will cycle among two states.

Instead, we observe that in the substitutes domain, excluding Cobb-Douglas utilities,
this rule is the mirror descent updating
rule using the KL divergence for the following optimization problem.

\begin{align*}
& \min_{\bbb} \hsp \Phi(\bbb) = - \sum_{ij} \frac{\bij}{\rhoi} 
 \log \frac{\aij (\bij)^{\rhoi - 1} } 
 {(\sum_{h}  b_{hj})^{\rho_i} } \\
&
\text{subject to}~\hhsp \textstyle{\sum_j} \bij = \ei \hsp \text{for all } i,
\hsp\text{and}~~\bij \geq 0~~~~\text{for all } i, j.
\end{align*}
We exclude Cobb-Douglas utilities, because as $\rhoi \ra 0$ the corresponding term
in $\Phi$ tends to $\infty$.
When restricted to linear utilities, i.e.\ $\rho_i = 1$ for all $i$, this is simply Shmyrev's convex program \cite{shmyrev2009algorithm} for these markets.

In the complementary domain, the mirror descent updating rule for this function is:
\begin{align}
& b_{ij}^{t+1} = \ei\cdot \frac { \Big(\frac{\aij} { {(\pjt)}^{\rhoi}}\Big)^\frac{1}{1 - \rhoi} }
 {\sum_k   \Big(\frac{\aik}{ {(\pkt)}^{\rhoi}}\Big)^\frac{1}{1 - \rhoi} }
\hsp\text{for}~-\infty < \rhoi < 0,
\hsp\hsp \text{and}~~~~b_{ij}^{t+1} = \ei \cdot \frac { \cij \pjt} {\sum_k c_{ik} \pkt}
\hsp\text{for}~ \rhoi = -\infty,\nonumber\\
&\text{where } \pkt = \textstyle{\sum_i} \bikt.\label{eqn:comp-CES-PR-rule}
\end{align}
Accordingly, we adopt this as the generalization of Proportional Response to the complementary domain.
This rule can be easily implemented in the distributed environment of Fisher markets.
In each round, each buyer only needs the prices computed in the previous round to compute its update.
Thus, to implement the update rule, it suffices to have the sellers broadcast their prices by the end of each round.

Interestingly, this update rule is also the best response action to the current prices for each buyer. 
We note that this rule can be viewed as a tatonnement update if we define
$x_{ij}^{t+1} = b_{ij}^{t+1}/\pjt$,
for then $p_j^{t+1} = \sum_i b_{ij}^{t+1} = \sum_i x_{ij}^{t+1} \pjt = \pjt(1 + z_{j}^{t+1})$. 
However, this is not the same rule as was used for
the tatonnement analyzed in recent works regarding Fisher markets~\cite{CF2008,CCD2013}.

For linear utilities, update rule~\eqref{eqn:subst-CES-PR-rule} was analyzed in~\cite{BDX2011}.
For the substitutes domain excluding linear utilities, a faster linear rate
of convergence was shown in~\cite{Zhang2011}, but not based
on considering the above optimization problem.
To obtain a linear rate of convergence for an analysis based on optimizing
$\Phi$ via a mirror descent
with a KL Divergence, we introduce the notion of \emph{strong Bregman convexity}.
We also coin the term Bregman convexity for an analogous notion introduced in~\cite{BDX2011}).

\begin{defn}
\label{def:LBreg}
The function $f$ is $L$-Bregman convex w.r.t.~Bregman divergence $d_h$ if, 
for any $\bby\in \rint(C)$ and $\bbx\in C$,
$$
f(\bby) + \inner{\nabla f(\bby)}{\bbx-\bby} ~~\le~~ f(\bbx) ~~\le~~ f(\bby) + \inner{\nabla f(\bby)}{\bbx-\bby} ~+~ L \cdot d_h(\bbx,\bby).
$$
The function $f$ is $(\sigma,L)$-strongly Bregman convex w.r.t.~Bregman divergence $d_h$ if,
$0<\sigma \le L$, and,
for any $\bby\in \rint(C)$, $\bbx\in C$,
$$
f(\bby) + \inner{\nabla f(\bby)}{\bbx-\bby} ~+~ \sigma \cdot d_h(\bbx,\bby) ~~\le~~ f(\bbx) ~~\le~~ f(\bby) + \inner{\nabla f(\bby)}{\bbx-\bby} ~+~ L \cdot d_h(\bbx,\bby).
$$
If the direction of the inequalities and the signs on the $ d_h(\bbx,\bby)$ terms are reversed, the function is said to be
Bregman concave (or strongly Bregman concave respectively).
($\rint(C)$ denotes the relative interior of $C$.)
\end{defn}
Consider the update rule:
\begin{equation}
\label{eq:md-update-rule}
\bbx^{t+1} ~~\la~~ \argmin_{\bby} \left\{ \inner{\nabla f(\bbx^t)}{\bby-\bbx^t} ~+~ L\cdot d_h(\bby,\bbx^t) \right\}.
\end{equation}

\noindent\begin{theorem}
\label{thm:md-convergence-rate}
Suppose that $f$ is $(\sigma,L)$-strongly Bregman convex w.r.t.~$d_h$.
If update rule \eqref{eq:md-update-rule} is applied, then,
for all $t\ge 1$,
$$
f(\bbx^t) - f(\bbx^*) ~~\le~~ \frac{\sigma}{\left(\frac{L}{L-\sigma}\right)^t - 1} \cdot d_h(\bbx^*,\bbx^0).
$$
\end{theorem}

An analogous Theorem for $L$-Bregman convex functions was given in \cite{BDX2011}:
\begin{theorem}{\cite{BDX2011}} \label{thm::plain::convex}
Suppose $f$ is an $L$-Bregman convex function w.r.t.~$d$, and $\bbx^T$ is the point reached after $T$ applications of the mirror descent update rule \eqref{eq:md-update-rule}. Then,
\begin{align*}
f(\bbx^T) - f(\bbx^*) \leq \frac{L \cdot d(\bbx^*, \bbx^0)}{T}.
\end{align*}
\end{theorem}

We show that in the CES substitutes domain, excluding linear utilities,
$\Phi$ is a strongly Bregman convex function
w.r.t.~the KL-divergence on spending,
thereby providing an alternative derivation of Zhang's result.
In addition, in the CES complements domain, excluding Leontief utilities,
$\Phi$ is a strongly Bregman concave function
w.r.t.~the KL divergence on spending,
which yields a proof of linear convergence in this domain.
These analyses are readily modified to give a $1/T$ rate of convergence
if we include respectively the linear and Leontief utilities.
These results are made precise in the following theorem.

\begin{theorem}
\label{thm:PRconvergence}
Suppose buyers 
with substitutes utilities
repeatedly update their spending using  
Proportional Response rule~\eqref{eqn:subst-CES-PR-rule},
and those with complementary utilities use rule~\eqref{eqn:comp-CES-PR-rule}.
Then the potential function $\Phi$ converges to the market equilibrium as follows.
\begin{itemize}
\item
If all buyers have substitutes CES utilities, then
\[
\Phi(\mathbf{b}^{T}) - \Phi(\mathbf{b}^*) \leq \frac{1}{T} \sum_i \frac{1}{\rho_i} \KL(b_i^* || b_i^0).
\]
\item
Suppose that in addition no buyer has a linear utility.
Let $\sigma = \min_i\{1 - \rho_i\}$. Then,
\begin{align*}
\Phi(\mathbf{b}^{T}) - \Phi(\mathbf{b}^*) \leq \frac{\sigma (1 - \sigma)^T}{1 - (1 - \sigma)^T}  \sum_i \frac{1}{\rho_i} \KL(b_i^* || b_i^0).
\end{align*}
\item
If all buyers have complementary CES utilities, then\footnote{For $\rhoi = - \infty$, we define $\frac{\rhoi - 1} {\rhoi}$ to equal 1.}
\begin{align*}
\Phi(\mathbf{b}^*) - \Phi(\mathbf{b}^{T}) \leq \frac{1}{T} \sum_i \frac{\rho_i - 1}{\rho_i} \KL(b_i^* || b_i^0).
\end{align*}
\item
Suppose that in addition no buyer has a Leontief utility.
Let $\sigma = \min_i\left\{\frac{1}{1 - \rho_i}\right\}$. Then,
\begin{align*}
\Phi(\mathbf{b}^*) - \Phi(\mathbf{b}^{T})\leq \frac{\sigma (1 - \sigma)^T}{1 - (1 - \sigma)^T}  \sum_i \frac{\rho_i - 1}{\rho_i} \KL(b_i^* || b_i^0).
\end{align*}
\end{itemize}
\end{theorem}
The  results are shown in Corollaries~\ref{cor:subst-incl-linear},~\ref{cor:subst-no-linear},~\ref{cor:comp-incl-Leontief}, and~\ref{cor:comp-no-Leontief}, respectively.
We also note that as shown in Lemma~13 in~\cite{BDX2011}, if $\bij = e_i/m$ for all $i$ and $j$,
then $\KL(\mathbf{b}^* || \mathbf{b}) \le \log mn$, which provides a possibly more intuitive version of the above bounds.

Theorem~\ref{thm:PRconvergence} does not cover buyers with Cobb-Douglas utilities,
because, as already noted, the terms in $\Phi$ for such buyers are equal to $\infty$.
Note that these buyers always wish to allocate their spending in fixed proportions regardless
of the prices.
Thus, arguably, it would be natural for these buyers to always have the equilibrium spending.
But even if this were not true initially, after one update this property would hold, and remain true
henceforth. Thus the presence of these buyers would seem to have little effect on the convergence.
Indeed, the above bounds hold with $\KL(\bbb_i^*||\bbb_i^0)$ replaced by $\KL(\bbb_i^*||\bbb_i^1)$ and $T$ replaced by $T-1$ on the RHS.
But to obtain bounds in terms of  $\KL(\bbb_i^*||\bbb_i^0)$ appears to require substantially more effort;
this analysis is given in Appendix~\ref{sec::entire_range}(see Theorems~\ref{thm::pure::sub::cobb}--\ref{thm::linear::cobb::com}).
The rates of convergence are similar to those given in Theorem~\ref{thm:PRconvergence}.

\subsection{Damped Proportional Response}
\label{sec:dampedPR}

But what if we want to allow a mix of substitutes and complementary utilities?
The difficulty we face is that the objective function $\Phi$ is no longer either convex or
concave. Rather, if we fix the spending of the buyers with complementary utilities,
the resulting restricted $\Phi$ is convex, while if we fix the spending of
the buyers with substitutes utilities,
the resulting restricted $\Phi$ is concave.
As it happens, the equilibrium corresponds to a saddle point of the function $\Phi$.
Also, a suitable dynamic will converge to this saddle point.
To show this, we introduce a saddle-point convergence analysis.
To this end, we define the following notion.

\begin{defn}\label{defn::convex::concave}
Function $f$ is $(L_X, L_Y)$-convex-concave
w.r.t.\ the pair of Bregman divergences $(\dg,d_h)$,
if it satisfies the following constraints.
\begin{enumerate}
\item For fixed $\bby$, $f(\cdot,\bby)$ is a convex function;
\item For fixed $\bbx$, $f(\bbx,\cdot)$ is a concave function;
\item There exist parameters $L_X, L_Y >0$ such that
for any $\bbx \in X$, $\bbx' \in X$, $\bby \in Y$ and $\bby' \in Y$, 
\begin{align*} \label{defn::saddle}
- L_Y \cdot d_h(\bby,\bby') \overset{(a)}{\leq} f(\bbx,\bby) - f(\bbx',\bby') - \langle \nabla f(\bbx',\bby'), (\bbx,\bby) - (\bbx', \bby') \rangle \overset{(b)}{\leq} L_X \cdot d_g(\bbx,\bbx'). \numberthis
\end{align*}
\end{enumerate}
\end{defn}
The saddle point is the ``optimal'' point of the convex-concave function, which is the minimum point in the $x$-direction and the maximum point in the $y$-direction, defined formally as follows.
\begin{defn}
$(x^*,y^*)$ is a saddle point of $f$ if and only if 
\begin{align*}
f(x, y^*) \geq f(x^*, y^*) \geq f(x^*, y) ~~~~~ \mbox{for any $x \in X$ and $y \in Y$.}
\end{align*}
\end{defn}

Now consider the following update rule:
\begin{align}
\nonumber
\bbx^{t+1} &= \arg\min_\bbx \{\langle \nabla_x f(\bbx^t, \bby^t), \bbx - \bbx^t \rangle + 2 L_X \cdot d_g(\bbx, \bbx^t)\}; \\
\label{eqn:sand-update}
\bby^{t+1} &= \arg\min_\bby \{\langle -\nabla_\bby f(\bbx^t, \bby^t), \bby - \bby^t \rangle + 2 L_Y \cdot d_h(\bby, \bby^t)\}.
\end{align}
We can then show an $O(1/T)$ empirical rate of convergence, as stated in the next theorem.
\begin{theorem}
\label{thm:sadd::conv::sublinear}
Suppose that $f$ is $(L_X, L_Y)$-convex-concave, and there exists a saddle point $(\bbx^*, \bby^*)$.
In addition, suppose that $(\bbx,\bby)$ is updated according to \eqref{eqn:sand-update}.
Then:
\begin{align*}
\text{(i)}~~~~\sum_{t = 1}^T \Big(f(\bbx^{t},\bby^*) - f(\bbx^*, \bby^{t})\Big) \leq 2L_X \cdot d_g(\bbx^*, \bbx^0) + 2L_Y \cdot d_h(\bby^*, \bby^0).\hspace*{3in}
\end{align*}
Note that  $f(\bbx^t, \bby^*) - f(\bbx^*, \bby^t) \geq 0$ since $f(\bbx^t, \bby^*) \geq f(\bbx^*, \bby^*) \geq f(\bbx^*, \bby^t)$. 

\noindent
(ii) Also, if  $\bar{\bbx} = \frac{1}{T} \sum_{t = 1}^{T} \bbx^t$ and $\bar{\bby} = \frac{1}{T} \sum_{t = 1}^{T} \bby^t$, then:
\begin{align*}
f(\bar{\bbx}, \bby^*) - f(\bbx^*, \bar{\bby}) \leq \frac{1}{T} [2 L_X \cdot d_g(\bbx^*, \bbx^0) + 2 L_Y \cdot d_h(\bby^*, \bby^0)].
\end{align*}
\end{theorem}
Note that the second part of the theorem follows immediately from the first part because $f(\cdot,\bby^*)$ is a convex function and  $f(\bbx^*,\cdot)$ is a concave function.

The objective function $\Phi$ is $(1, 1)$-convex-concave w.r.t.\ $\dg=\sum_{i: \rho_i > 0} \frac{1}{\rho_i} \KL(b_i || b_i')$ and $d_h= \sum_{i: \infty < \rho_i < 0} \frac{\rho_i - 1}{\rho_i} \KL(b_i || b_i') + \sum_{i: \rho_i = -\infty} \KL(b_i || b_i')$.
Consequently, we obtain an empirical $O(1/T)$ rate of convergence for the following \emph{Damped} Proportional Response update.
\begin{align*}
b_{ij}^{t+1} &= e_i \cdot \frac{\Big[b_{ij}^t \cdot a_{ij} \Big(\frac{b^t_{ij}}{p^t_j} \Big)^{\rho_i}\Big]^{\frac{1}{2}}}{\sum_{k}\Big[b_{ik}^t \cdot a_{ik} \Big( \frac{b^t_{ik}}{p^t_{k}} \Big)^{\rho_i}\Big]^{\frac{1}{2}}}, \hhsp\mbox{for $\rho_i > 0$;} 
\hsp b_{ij}^{t+1} = e_i \cdot \frac{\Big[b_{ij}^t \cdot \Big(\frac{a_{ij}}{(p^t_{j})^{\rho_i}}\Big)^{\frac{1}{1 - \rho_i}}\Big]^{\frac{1}{2}}}{\sum_{k}\Big[b_{ik}^t \cdot \Big(\frac{a_{ik}}{(p^t_{k})^{\rho_i}}\Big)^{\frac{1}{1 - \rho_i}}\Big]^{\frac{1}{2}}},\hhsp\mbox{for $-\infty < \rho_i < 0$;}\\
b_{ij}^{t+1} &= e_i \cdot \frac{\Big[b_{ij}^t \cdot \Big(\frac{c_{ij}}{p^t_{j}}\Big)^{-1}\Big]^{\frac{1}{2}}}{\sum_{k}\Big[b_{ik}^t \cdot \Big(\frac{c_{ik}}{p^t_{k}}\Big)^{-1}\Big]^{\frac{1}{2}}},\hhsp\mbox{for $\rho_i = -\infty$;}\numberthis \label{update::damped}
\end{align*}
We say it is damped because the update rule uses the geometric mean
of the current value and the standard Proportional Response update.

A natural question is whether a linear convergence rate is possible if the
linear and Leontief utilities are excluded.
The answer is yes, and to obtain this we need a stronger condition on the convex-concave
objective function, as given in the following definition.
\begin{defn}\label{defn::strong::saddle}
$f$ is a $(\sigma_X, \sigma_Y, L_X, L_Y)$-strongly Bregman convex-concave function,
w.r.t.\ Bregman divergences $\dg, d_h$, if,
for all $\bbx \in X$, $\bbx' \in X$, $\bby \in Y$, and $\bby' \in Y$, function $f$ satisfies:
\begin{align*}\label{defn::strong::saddle::eqn}
&-L_Y \cdot d_h(\bby,\bby')  +  \sigma_X \cdot d_g(\bbx, \bbx')
\leq f(\bbx,\bby) - f(\bbx',\bby') - \langle \nabla f(\bbx',\bby'), (\bbx,\bby) - (\bbx', \bby') \rangle \hspace*{1in} \\
&\hspace*{3.75in}\leq L_X \cdot  d_g(\bbx,\bbx') - \sigma_Y \cdot d_h(\bby, \bby'). \numberthis
\end{align*} 
\end{defn}

\noindent\begin{theorem}
\label{conv::saddle::linear}
If $f$ is a $(\sigma_X, \sigma_Y, L_X, L_Y)$-strongly Bregman convex-concave function w.r.t $\dg$ and $d_h$, 
and there exists a saddle point, then update rule~\eqref{eqn:sand-update}
converges to the saddle point with a linear convergence rate:
\begin{align*}
\Big(f(\bbx^{T}, \bby^*) - f(\bbx^*, \bby^{T})\Big) \leq \Bigg(1 - \frac{\min\left\{\frac{\sigma_X}{L_X}, \frac{\sigma_Y}{L_Y}\right\}}{2}\Bigg)^{T-1} \Big( (2 L_X - \sigma_X)  d_g(\bbx^*, \bbx^0) + (2 L_Y - \sigma_Y) d_h(\bby^*, \bby^0)\Big).
\end{align*}
\end{theorem}
$\Phi$ is $(\min_{i:\rho_i > 0} \{1 - \rho_i\}, \min_{i: \rho_i < 0} \left\{\frac{1}{1 - \rho_i}\right\},1,1)$-strongly Bregman convex-concave, and thus we
can deduce that the Damped Proportional Response achieves a linear
convergence rate if linear and Leontief utilities are excluded.

As before, the above results exclude Cobb-Douglas utilities.

Arguably, the buyers with Cobb-Douglas utilities should always have the equilibrium spending, or failing that, should
immediately update to this spending.
But for mathematical consistency, we suppose they are performing the same type of damped update as the other buyers.
In this case, our previous potential function can't be used when we include Cobb-Douglas utility functions.
We now need to include a term in the potential function
for each buyer with a Cobb-Douglas utility as their spending keeps changing.

We will need the following notation. 
Let $\mathbf{b}_{>0}$, $\mathbf{b}_{=0}$, and $\mathbf{b}_{<0}$ denote the spending of those buyers with
 $\rho_i > 0$, $\rho_i = 0$, and $\rho_i < 0$, respectively. 
Accordingly, we will write $\Phi(\mathbf{b}) = \Phi(\mathbf{b}_{>0}, \mathbf{b}_{=0}, \mathbf{b}_{<0})$.
The resulting function is still convex in $\bbb_{>0}$ and concave in $\bbb_{<0}$.
The construction is given in Appendix~\ref{construction::cobb}.
We note that the update rule for the buyers with $\rhoi=0$ is given by:
\[
b_{ij}^{t+1} = e_i \cdot \frac{\Big[b_{ij}^t \cdot a_{ij}^t\Big]^{\frac{1}{2}}}{\sum_{k}\Big[b_{ik}^t \cdot a_{ik}^t\Big]^{\frac{1}{2}}},\hsp\hsp\mbox{for $\rho_i = 0$.}
\]

\begin{theorem}
Suppose buyers repeatedly update their spending using the 
Damped Proportional Response rule~\eqref{update::damped}.
Then 
\begin{align*}
\KL(\mathbf{b}^{T}_{=0} || \mathbf{b}^*_{=0}) \leq \frac{1}{2^T} \KL(\mathbf{b}^{0}_{=0} || \mathbf{b}^*_{=0}),
\end{align*}
and the potential function $\Phi$ converges to the market equilibrium as follows:

\begin{align*}
&\text{i.}\hsp \sum_{t = 1}^{T} \Bigg[\Phi(\mathbf{b}^{t}_{>0}, \mathbf{b}^{*}_{=0}, \mathbf{b}^{*}_{<0}) - \Phi (\mathbf{b}^{*}_{>0}, \mathbf{b}^{*}_{=0}, \mathbf{b}^{t}_{<0}) \Bigg] \hspace*{3in}\\
 &\hspace*{0.2in}\hsp\hsp \leq 4 \sum_{i : \rho_i = 0} \KL(\bbb^*_i || \bbb^0_i)  + \sum_{i: -\infty < \rho_i < 0} \frac{2(\rho_i - 1)}{\rho_i}  \KL(\bbb_i^* || \bbb_i^0) + \sum_{i: \rho_i > 0} \frac{2}{\rho_i}  \KL(\bbb_i^* || \bbb_i^0)+ \sum_{i: \rho_i = -\infty} \KL(\bbb_i^* || \bbb_i^0).\\
&\text{ii. If in addition no buyer has a linear or Leontief utility, }\\
&\hspace*{1.0in}\text{Let }~\sigma = \min\left\{\hhsp\min_{i:\rho_i > 0} \left\{\frac{2}{1+\rho_i}\right\}\hhsp,\hhsp\min_{i:\rho_i < 0} \left\{\frac{2(\rho_i - 1)}{2\rho_i - 1}\right\}\hhsp\right\}\hsp\hsp\mbox{(so $1 < \sigma < 2$).}\hhsp\text{Then,}
\end{align*}
\begin{align*}
 \Phi(\mathbf{b}^{T}_{>0}, \mathbf{b}^{*}_{=0}, \mathbf{b}^{*}_{<0}) - \Phi (\mathbf{b}^{*}_{>0}, \mathbf{b}^{*}_{=0}, \mathbf{b}^{T}_{<0}) &\leq \frac{1}{\sigma^{T-1}} \Bigg[ \frac{4}{2 - \sigma}   \sum_{i: \rho_i = 0} \KL(\bbb_i^* || \bbb_i^0) \\
&~~~~~~~~~+ \sum_{i: \rho_i < 0} \frac{2\rho_i - 1}{\rho_i}  \KL(\bbb_i^* || \bbb_i^0)  + \sum_{i: \rho_i > 0} \frac{1+\rho_i}{\rho_i}  \KL(\bbb_i^* || \bbb_i^0)\Bigg]. 
\end{align*} 
\end{theorem}
Weaker bounds are shown in Corollaries~\ref{cor:full-range-no-constraints} and~\ref{cor:no-extremes}, respectively. The complete proof is given in  
Appendix~\ref{sec::entire_range}   (see Theorems~\ref{thm::complete::1T} and~\ref{thm::complete::linear}).
\section{Related Work}\label{sec:related-work}

The concept of a market equilibrium was first proposed by Walras~\cite{Walras1874}
along with a description of the tatonnement process.
Since then, studies of market equilibria and tatonnement have received much attention in
economics, operations research, and most recently in computer science.
A fairly recent account of the classic perspective in economics is given in~\cite{McK2002}.

Computer scientists, beginning with the work by Deng et al.~\cite{papadeng02},
showed that computing equilibria was a hard problem in general; see also~\cite{DengDu2008,PY2010}.
This led to much work on polynomial time algorithms for restricted classes of markets,
e.g.~\cite{DPSV08,DV04,CMV2005,GargKap2006}.

The Eisenberg-Gale program for the case of linear utilities
was formulated in~\cite{EisenbergGale} and then
generalized to homothetic functions in~\cite{Eisenberg61}; further generalizations were given in~\cite{JainVaz07}.
The maxima of these convex programs correspond to the equilibria of the corresponding markets.
In particular, when buyer or agent utilities are homothetic,
the optimum of the Eisenberg-Gale program corresponds to the optimum Nash Social Welfare;
interestingly, this optimum also appears to provide good outcomes when apportioning indivisible goods~\cite{Budisch11,CMP016}.
Recently, Cole et al.~\cite{CDGJMVY2017} identified another variant of the Eisenberg-Gale program that captured
the best currently-known polynomial-time approximate solution for the indivisible setting.

The analysis most similar to ours is the one in~\cite{BDX2011} which considers convex functions that obey
a constraint which we name $L$-Bregman convexity w.r.t.\ a Bregman divergence (see Definition~\ref{def:LBreg}).
Our work generalizes this notion substantially.

The earliest analyses of tatonnement showed convergence in exchange economies with gross substitutes utilities,
first for continuous updating~\cite{ABH1959} and then for discrete updates~\cite{Uzawa1960},
but it was shown to diverge in general~\cite{scarf60}. 
Recent works have analyzed its convergence properties in specific markets, primarily Fisher markets~\cite{CF2008,CCR2012,CCD2013}.
Cheung et al.~\cite{CCD2013} showed that tatonnement is equivalent to coordinate descent
on a convex function for several classes of Fisher markets, and consequently
that a suitable tatonnement converges toward the market equilibrium in two classes of markets:
complementary-CES Fisher markets and Leontief Fisher markets.

Other dynamics have been considered.
In particular, Dvijotham et al.~\cite{DRS2016} study sellers best responding in a setting in which they
form beliefs about other sellers' strategies.
They obtain linear convergence in Fisher markets for most of the CES domain, but not for linear utilities.
In the context of network flow control, Low and Lapsley~\cite{low1999optimization} adopted an optimization approach to derive a dynamic protocol where
both prices (of links) and flow demands of agents are updated, and showed that the protocol converges to a social-welfare maximizing state.
The update rules \eqref{eqn:subst-CES-PR-rule}, \eqref{eqn:comp-CES-PR-rule}
look quite similar to a game-learning dynamic called \emph{log-linear learning}~\cite{Blume1993,MS2012}
(by suitably viewing spendings as probability densities), but due to different contexts (games vs.~markets),
the actual behaviors and the analyses have significant qualitative differences.

Convex-concave saddle-point problems can be reduced to non-smooth convex minimization problems,
for which algorithms yielding $O(1/\sqrt{T})$ convergence rate exist.
Its wide applications (e.g., to two-person zero-sum game equilibria) 
have motivated exploration of properties of the underlying function
which support faster converging algorithms~\cite{Nemirovski2004,Nesterov2005a,Nesterov2005b,nesterov2007dual,RS2013}.
In this paper, we present a new property and a simple algorithm which yields an $O(1/T)$ empirical convergence rate.
In our opinion, its analysis is quite simple, which may well open the door to further exploration.
Indeed, we have taken such a step by presenting a variant of our new property for which the same algorithm yields
a linear \emph{point-wise} convergence rate.

\section{Linear Convergence with Strong Bregman Convexity}\label{sec:strong-Breg}

Our proof will use the following lemmas.
\begin{lemma}{\cite{CT1993}}\label{ineq::basic::bregman}
If $\bbx^{+}$ is the optimal point for the optimization problem:
\begin{align*}
&\mbox{\text{minimize}} \hspace*{0.2in}g(\bbx) + d(\bbx,\bby) \\\
&\mbox{\text{subject to}}\hspace*{0.2in}\bbx \in C,
\end{align*}
where $C$ is a compact convex set, then, for any $\bbx \in C$, 
\begin{align*}
g(\bbx^{+}) + d(\bbx^{+}, \bby) + d(\bbx, \bbx^{+}) \leq g(\bbx) + d(\bbx, \bby).
\end{align*}
\end{lemma}
\begin{lemma}{\cite{BDX2011}}\label{lem::next_smaller}
Suppose that $f$ is an $L$-Bregman convex function w.r.t.~$d(\bbx,\bbx')$, and $\bbx^t$ and $\bbx^{t+1}$ are the points reached after $t$ and $t+1$ applications of the mirror descent update rule \eqref{eq:md-update-rule}. Then
\begin{align*}
f(\bbx^{t+1}) \leq f(\bbx^t).
\end{align*}
\end{lemma}
\begin{proof}[Proof of Theorem \ref{thm:md-convergence-rate}]

By Lemma~\ref{ineq::basic::bregman} with $\bby = \bbx^t$, $\bbx^+ = \bbx^{t+1}$, and $\bbx = \bbx^*$,
\begin{align*}\label{ineq::strong_convex::1}
\inner{\nabla f(\bbx^t)}{\bbx^{t+1} - \bbx^t} ~+~ L \cdot d_h(\bbx^{t+1},\bbx^t) ~~\le~~ \inner{\nabla f(\bbx^t)}{\bbx^* - \bbx^t}
~+~ L \cdot \left[ d_h(\bbx^*,\bbx^t) - d_h(\bbx^*,\bbx^{t+1}) \right]. \numberthis
\end{align*}
By strong Bregman-convexity, with $\bby = \bbx^t$, and $\bbx = \bbx^{t+1}$,
\begin{align*}\label{ineq::strong_convex::2}
\inner{\nabla f(\bbx^t)}{\bbx^{t+1} - \bbx^t} ~+~ L \cdot d_h(\bbx^{t+1},\bbx^t) ~~\ge~~ f(\bbx^{t+1}) - f(\bbx^t); \numberthis
\end{align*}
and with $\bby = \bbx^t$ and $\bbx = \bbx^*$,
\begin{align*}\label{ineq::strong_convex::3}
\nabla f(\bbx^t) \cdot (\bbx^* - \bbx^t) ~~\le~~ f(\bbx^*) - f(\bbx^t) - \sigma \cdot d_h(\bbx^*,\bbx^t). \numberthis
\end{align*}
Combining \eqref{ineq::strong_convex::1}, \eqref{ineq::strong_convex::2}, and \eqref{ineq::strong_convex::3}, gives, for $t\ge 0$,
\begin{align}
\label{eqn::comb-strong-convex}
f(\bbx^{t+1}) - f(\bbx^*) ~~\le~~ (L-\sigma) \cdot d_h(\bbx^*,\bbx^t) ~-~ L \cdot d_h(\bbx^*,\bbx^{t+1}).
\end{align}

On multiplying both sides of the above inequality by $\left( \frac{L}{L-\sigma} \right)^t$, and then summing over $0\le t < T$,
the RHS becomes a telescoping sum, and hence
$$
\sum_{t=0}^{T-1} \left( \frac{L}{L-\sigma} \right)^t \cdot \left[ f(\bbx^{t+1}) - f(\bbx^*) \right]
~~\le~~ (L-\sigma) \cdot d_h(\bbx^*,\bbx^0).
$$
By Lemma~\ref{lem::next_smaller}, $f(\bbx^{t+1})\le f(\bbx^t)$; thus:
$$
\frac{L-\sigma}{\sigma} \cdot \left[ \left(\frac{L}{L-\sigma}\right)^T - 1 \right] \cdot \left[ f(\bbx^T) - f(\bbx^*) \right] ~~=~~ \left( \sum_{t=0}^{T-1} \left( \frac{L}{L-\sigma} \right)^t \right) \cdot \left[ f(\bbx^T) - f(\bbx^*) \right] ~~\le~~ (L-\sigma) \cdot d_h(\bbx^*,\bbx^0),
$$
and the result follows.
\end{proof}

\section{Convergence of Proportional Response}

We consider the following potential function:
\begin{align*}
&p_j(\mathbf{b}) = \textstyle{\sum_{i}} b_{ij}, \\
&\Phi(\mathbf{b}) = - \sum_{i: \rho_i \neq \{0, -\infty\}} \frac{1}{\rho_i} \sum_{j} b_{ij} \log \frac{a_{ij} b_{ij}^{\rho_i - 1}}{[p_j(\mathbf{b})]^{\rho_i}} - \sum_{i: \rho_i = -\infty} \sum_j b_{ij} \log \frac{b_{ij}}{c_{ij}p_j(\mathbf{b})}.
\end{align*}

For those $i$ for which $\rho_i \neq -\infty$,
\begin{align*}
\nabla_{b_{ij}} \Phi(\mathbf{b}) &= - \frac{1}{\rho_i}\log a_{ij} - \frac{\rho_i - 1}{\rho_i}(\log b_{ij} + 1) + \log p_j(\mathbf{b}) + \sum_{h}  b_{hj} \frac{1}{p_j(\mathbf{b})} \\
&= \frac{1}{\rho_i}\Big(1 - \log \frac{a_{ij} b_{ij}^{\rho_i - 1}}{p_j^{\rho_i}(\mathbf{b})}\Big);
\end{align*}
\begin{align*}
\mbox{and for those $i$ for which $\rho_i = -\infty$,} \hspace*{0.2in}
\nabla_{b_{ij}} \Phi(\mathbf{b}) = - \log \frac{ b_{ij}}{c_{ij} p_j}. \hspace*{5in}
\end{align*}

We deduce:
\begin{lemma} \label{ineq::saddle::upper::lower}
\begin{align*}
\sum_{i: \rho_i \neq -\infty} \frac{1 - \rho_i}{\rho_i} \KL(\bbb_i || \bbb_i') - \sum_{i: \rho_i = -\infty} \KL(\bbb_i || \bbb_i') \leq \Phi(\bbb) - \Phi(\bbb') - \langle \nabla \Phi(\mathbf{b}'), \mathbf{b} - \mathbf{b}' \rangle \leq \sum_{i: \rho_i \neq  -\infty} \frac{1}{\rho_i} \KL(\bbb_i || \bbb_i'). 
\end{align*}
\end{lemma}
\begin{proof}
\begin{align*}
\Phi(\mathbf{b}) - \Phi(\mathbf{b}') - \langle \nabla \Phi(\mathbf{b}'), \mathbf{b} - \mathbf{b}' \rangle &= -\sum_{i: \rho_i \neq \{0, -\infty\}} \frac{\rho_i - 1}{\rho_i} \KL(\bbb_i || \bbb_i') - \sum_{i:\rho_i = -\infty} \KL(\bbb_i || \bbb_{i}')  + \KL(\bbp || \bbp') \\
&= \sum_{i: \rho_i \neq  -\infty} \frac{1}{\rho_i} \KL(\bbb_i || \bbb_i')  - \Big( \sum_i \KL(\bbb_i || \bbb_i') - \KL(\bbp || \bbp')\Big).
\end{align*}

Since $\sum_i \KL(\bbb_i || \bbb_i') \geq \KL(\bbp || \bbp')$, the result follows.
\end{proof}

\paragraph{The Substitutes Domain}
The following lemma states the equivalence between mirror descent and Proportional Response in the substitutes domain; it follow readily from the definition of $\bbb^{t+1}$ for Proportional Response (given by \eqref{eqn:subst-CES-PR-rule}).
\begin{lemma}
For buyers with CES substitutes utilities, the Proportional Response update is the same as the mirror descent update, given by:
\begin{align*}
\bbb_{i}^{t+1} = {\arg \min}_{\bbb_i : \sum_j b_{ij} = e_i} \left\{\langle \nabla_{\bbb_i}\Phi(\mathbf{b}^t), \bbb_i - \bbb_i^t\rangle + \frac{1}{\rho_i} \KL(\bbb_i || \bbb_i^t)\right\}.
\end{align*}
\end{lemma}

The next lemma states several properties of the potential function in the substitutes domain.
\begin{lemma}
\begin{enumerate}[i.]
\item If $\rho_i > 0$ for all $i$, then $\Phi$ is a $1$-Bregman convex function w.r.t. $\sum_i \frac{1}{\rho_i} \KL(\bbb_i || \bbb'_i)$;
\item if $0 < \rho_i < 1$ for all $i$, then $\Phi$ is a $(\min_i\{1 - \rho_i\}, 1)$-strong Bregman convex function w.r.t. $\sum_i \frac{1}{\rho_i} \KL(\bbb_i || \bbb'_i)$;
\item $\mathbf{b}$ is the spending at the market equilibrium if and only if $\mathbf{b}$ is the minimum point of $\Phi$. 
\end{enumerate}
\end{lemma}

\begin{proof}
The first two claims follow from  Lemma~\ref{ineq::saddle::upper::lower} with a little calculation.
The proof of the third claim is given in Appendix~\ref{correspondence}.
\end{proof}

Let $\mathbf{b}^*$ be the spending at some market equilibrium. Applying Theorem~\ref{thm::plain::convex} yields:

\begin{corollary}\label{cor:subst-incl-linear}
\begin{align*}
\Phi(\mathbf{b}^{T}) - \Phi(\mathbf{b}^*) \leq \frac{1}{T} \sum_i \frac{1}{\rho_i} \KL(\bbb_i^* || \bbb_i^0).
\end{align*}
\end{corollary}
Furthermore, if there is no buyer with a linear utility function, then, applying Theorem~\ref{thm:md-convergence-rate} yields:
\begin{corollary} 
\label{cor:subst-no-linear}
Let $\sigma = \min_i\{1 - \rho_i\} > 0$. Then
\begin{align*}
\Phi(\mathbf{b}^{T}) - \Phi(\mathbf{b}^*) \leq \frac{\sigma (1 - \sigma)^T}{1 - (1 - \sigma)^T}  \sum_i \frac{1}{\rho_i} \KL(\bbb_i^* || \bbb_i^0).
\end{align*}
\end{corollary}

We now explain how to recover Zhang's bound~\cite{Zhang2011}.
From \eqref{eqn::comb-strong-convex},
\begin{align*}
\sum_i \frac{1}{\rho_i} \KL(\mathbf{b}_i^* || \mathbf{b}_i^t) \le \frac {L - \sigma}{L} \sum_i \frac{1}{\rho_i} \KL(\mathbf{b}_i^* || \mathbf{b}_i^{t-1})
\le \left(\frac {L - \sigma}{L}\right)^t  \sum_i \frac{1}{\rho_i} \KL(\mathbf{b}_i^* || \mathbf{b}_i^{0})= \left(\max_i \rho_i\right)^t \sum_i \frac{1}{\rho_i} \KL(\mathbf{b}_i^* || \mathbf{b}_i^{0}).
\end{align*}
In ~\cite{Zhang2011}, $\phi(t)$ is used to denote $\sum_i \frac{1}{\rho_i}\KL(\mathbf{b}_i^* || \mathbf{b}_i^t)$.
We have obtained the exact same bound on $\sum_i \frac{1}{\rho_i}\KL(\mathbf{b}_i^* || \mathbf{b}_i^t)$ as in~\cite{Zhang2011}, 
and thus can deduce the identical convergence rate.

\paragraph{The Complementary Domain}

We proceed as in the substitutes domain. First, the following lemma shows the equivalence between the mirror descent and the Proportional Response (given by \eqref{eqn:comp-CES-PR-rule}) in the complementary domain.
\begin{lemma}
For those complementary buyers such that $-\infty< \rho_i < 0$, the Proportional Response update, which is the best response in this domain, is equal to the mirror descent update, given by:
\begin{align*}
\bbb_{i}^{t+1} = {\arg \min}_{\bbb_i : \sum_j b_{ij} = e_i} \left\{- \langle \nabla_{b_i}\Phi(\mathbf{b}^t), \bbb_i - \bbb_i^t\rangle + \frac{\rho_i - 1}{\rho_i} \KL(\bbb_i || \bbb_i^t)\right\};
\end{align*}
and this also holds for buyers with Leontief utility functions, where now the mirror descent update is given by:

\begin{align*}
\bbb_{i}^{t+1} = {\arg \min}_{\bbb_i : \sum_j b_{ij} = e_i} \{- \langle \nabla_{b_i}\Phi(\mathbf{b}^t), \bbb_i - \bbb_i^t\rangle +  \KL(\bbb_i || \bbb_i^t)\};
\end{align*}
\end{lemma}

Next, we show the following  properties of the potential function in the complementary domain. 
The main difference between the complementary case and the substitutes case is that the potential function is a concave function in the complementary domain, while it is a convex function in the substitutes domain.
\begin{lemma}
\begin{enumerate}[i.]
\item If $\rho_i < 0$ for all $i$, then $\Phi$ is a $1$-Bregman concave function w.r.t. $\sum_i \frac{\rho_i - 1}{\rho_i} \KL(\bbb_i || \bbb'_i)$;
\item if $-\infty < \rho_i < 0$ for all $i$, then $\Phi$ is a $(\min_i\{\frac{1}{1 - \rho_i}\}, 1)$-strong Bregman concave function w.r.t. $\sum_i \frac{\rho_i - 1}{\rho_i} \KL(\bbb_i || \bbb'_i)$;
\item $\mathbf{b}$ is the spending at the market equilibrium if and only if $\mathbf{b}$ is the maximum point of $\Phi$. 
\end{enumerate}
\end{lemma}
\begin{proof}
The first two claims follow from  Lemma~\ref{ineq::saddle::upper::lower} with a little calculation.
The proof of the third claim is given in Appendix~\ref{correspondence}.
\end{proof}

Also, let $\mathbf{b}^*$ be the spending at some market equilibrium. Applying Theorem~\ref{thm::plain::convex} 
yields:\footnote{Recall that for $\rhoi=-\infty$, we defined $(\rhoi - 1)/ \rhoi = 1$.}

\begin{corollary}
\label{cor:comp-incl-Leontief}
\begin{align*}
\Phi(\mathbf{b}^*) - \Phi(\mathbf{b}^{T}) \leq \frac{1}{T} \sum_i \frac{\rho_i - 1}{\rho_i} \KL(\bbb_i^* || \bbb_i^0).
\end{align*}
\end{corollary}
In addition, if there is no buyer with a Leontief utility, applying Theorem~\ref{thm:md-convergence-rate} yields:
\begin{corollary} 
\label{cor:comp-no-Leontief}
Let $\sigma = \min_i\{\frac{1}{1 - \rho_i}\}$. Then,
\begin{align*}
\Phi(\mathbf{b}^*) - \Phi(\mathbf{b}^{T})\leq \frac{\sigma (1 - \sigma)^T}{1 - (1 - \sigma)^T}  \sum_i \frac{\rho_i - 1}{\rho_i} \KL(\bbb_i^* || \bbb_i^0).
\end{align*}
\end{corollary}

\section{Saddle Point Analysis}\label{sec:saddle-point}

\begin{proof}[Proof of Theorem~\ref{thm:sadd::conv::sublinear}]
Recall that $\bbx^{t+1} = {\arg\min}_{\bbx \in X} \{\langle \nabla_{\bbx} f(\bbx^t, \bby^t), \bbx - \bbx^t \rangle + 2 L_X d_g(\bbx, \bbx^t)\}$. Applying Lemma~\ref{ineq::basic::bregman} with $\bbx = \bbx^*$, $d(\cdot, \cdot) = 2L_X d_g(\cdot, \cdot)$, $\bby = \bbx^t$, $\bbx^+ = \bbx^{t+1}$, and $g(\bbx) = \langle \nabla_x f(\bbx^t, \bby^t) \rangle, \bbx - \bbx^t\rangle$ gives
\begin{align*}
\langle \nabla_x f(\bbx^t, \bby^t), \bbx^{t+1} - \bbx^t \rangle + 2 L_X d_g(\bbx^{t+1}, \bbx^t) \leq \langle \nabla_x f(\bbx^t, \bby^t), \bbx^* - \bbx^t \rangle + 2 L_X d_g(\bbx^*, \bbx^t) - 2 L_X d_g(\bbx^*, \bbx^{t+1}).
\end{align*}

This is equivalent to
\begin{align*} 
\label{sadd::ineq::comp::point}
&\underbrace{f(\bbx^t, \bby^t) + \langle \nabla f(\bbx^t, \bby^t), (\bbx^{t+1}, \bby^{t+1}) - (\bbx^t, \bby^t) \rangle + 2 L_X d_g(\bbx^{t+1}, \bbx^t)}_{\texttt{LHS}} \\
&\hspace*{0.2in}\leq \underbrace{f(\bbx^t, \bby^t) + \langle \nabla f(\bbx^t, \bby^t), (\bbx^*, \bby^{t+1}) - (\bbx^t, \bby^t) \rangle + 2 L_X d_g(\bbx^*, \bbx^t) - 2 L_X d_g(\bbx^*, \bbx^{t+1}).}_{\texttt{RHS}} \numberthis
\end{align*}

Since $f$ is $(L_X, L_y)$-convex-concave, the third property --- see \eqref{defn::saddle} --- gives:
\begin{align*} \label{sadd:ineq::pos}
&f(\bbx^{t+1}, \bby^{t+1}) + L_X d_g(\bbx^{t+1}, \bbx^t) \overset{(1)}{\leq} \texttt{LHS} \leq \texttt{RHS} \\
&\hspace*{0.6in}\overset{(2)}{\leq}  f(\bbx^*, \bby^{t+1}) + L_Y d_h(\bby^{t+1}, \bby^t) + 2 L_X d_g(\bbx^*, \bbx^t) - 2 L_X d_g(\bbx^*, \bbx^{t+1}), \numberthis
\end{align*}
where $(1)$ is deduced from $(b)$ in Definition~\ref{defn::convex::concave} with $(x', y') = (\bbx^t, \bby^t)$ and $(2)$ is deduced from $(a)$ in Definition~\ref{defn::convex::concave} with $(x', y') = (\bbx^t, \bby^t)$ and $(x, y) = (\bbx^*, \bby^{t+1})$.

Now, let's consider $-f(x,y)$ and $\bby^{t+1} = \arg\min_{\bby \in Y} \{\langle -\nabla_{\bby} f(\bbx^t, \bby^t), \bby - \bby^t \rangle + 2 L_Y d_h(\bby, \bby^t)\}$. Using a similar argument, we obtain:
\begin{align*}
\label{sadd:ineq::neg}
-f(\bbx^{t+1}, \bby^{t+1}) + L_Y d_h(\bby^{t+1}, \bby^t) \leq -f(\bbx^{t+1}, \bby^*) + L_X d_g(\bbx^{t+1}, \bbx^t) + 2 L_Y d_h(\bby^*, \bby^t) - 2 L_Y d_h(\bby^*, \bby^{t+1}) . \numberthis
\end{align*}

Adding these two inequalities gives:
\begin{align*}
f(\bbx^{t+1}, \bby^*) - f(\bbx^*, \bby^{t+1}) \leq 2 L_X d_g(\bbx^*, \bbx^t) + 2 L_Y d_h(\bby^*, \bby^t) - 2 L_X d_g(\bbx^*, \bbx^{t+1}) - 2 L_Y d_h(\bby^*, \bby^{t+1}).
\end{align*}

Summing over $t$ yields:\hspace*{0.2in}
$\sum_{t = 1}^T \Big(f(\bbx^{t}, \bby^*) - f(\bbx^*, \bby^{t})\Big) \leq 2 L_X d_g(\bbx^*, \bbx^0) + 2 L_Y d_h(\bby^*, \bby^0)$.
\end{proof}

\begin{proof}[Proof of Theorem~\ref{conv::saddle::linear}]
Using \eqref{defn::strong::saddle::eqn} instead of \eqref{defn::saddle}, we deduce the following from \eqref{sadd::ineq::comp::point} instead of \eqref{sadd:ineq::pos}:
\begin{align*}
f(\bbx^{t+1}, \bby^{t+1}) + L_X d_g(\bbx^{t+1}, \bbx^t) \leq f(\bbx^*, \bby^{t+1}) + L_Y d_h(\bby^{t+1}, \bby^t) + (2 L_X - \sigma_X) d_g(\bbx^*, \bbx^t) - 2 L_X d_g(\bbx^*, \bbx^{t+1}).
\end{align*}

Also,  \eqref{sadd:ineq::neg} is replaced by: 
\begin{align*}
-f(\bbx^{t+1}, \bby^{t+1}) + L_Y d_h(\bby^{t+1}, \bby^t) \leq -f(\bbx^{t+1}, \bby^*) + L_X d_g(\bbx^{t+1}, \bbx^t) + (2 L_Y - \sigma_Y) d_h(\bby^*, \bby^t) - 2 L_Y d_h(\bby^*, \bby^{t+1}).
\end{align*}

Summing up these two inequalities gives:
\begin{align*}
&f(\bbx^{t+1}, \bby^*) - f(\bbx^*, \bby^{t+1}) \\
&\hspace*{0.2in}\leq (2 L_X - \sigma_X) d_g(\bbx^*, \bbx^t) + (2 L_Y - \sigma_Y) d_h(\bby^*, \bby^t) - 2 L_X d_g(\bbx^*, \bbx^{t+1}) - 2 L_Y d_h(\bby^*, \bby^{t+1}).
\end{align*}

Let $\sigma = \min\left\{\frac{\sigma_X}{L_X},\frac{ \sigma_Y}{L_Y}\right\}$. Then:
\begin{align*}
\sum_{t = 0}^{T-1} \left(\frac{2}{2-\sigma}\right)^t \Big(f(\bbx^{t+1}, \bby^*) - f(\bbx^*, \bby^{t+1})\Big) \leq (2 L_X - \sigma_X) d_g(\bbx^*, \bbx^0) + (2 L_Y - \sigma_Y) d_h(\bby^*, \bby^0).
\end{align*}

Note that $f(\bbx^{t}, \bby^*) - f(\bbx^*, \bby^{t})$ is positive for each $t$, so the result follows.
\end{proof}

\section{Analysis of Damped Proportional Response}

\paragraph{Excluding Cobb-Douglas Utility Functions}
First, we consider a simplified situation where there is no buyer with a Cobb-Douglas utility function.
We want to use the technique developed in the saddle point analysis to obtain a convergence result.  
The potential function is the same as before.

We make the following observations.
\begin{lemma} \label{lem::HGPRD::eq}
If $\rho_i > 0$ for buyer $i$, then the Damped Proportional Response (given by \eqref{update::damped}) is equivalent to mirror descent with a halved step size, defined as follows:
\begin{align*}
\bbb_{i}^{t+1} = {\arg \min}_{\bbb_{i} : \sum_j b_{ij} = e_i} \left\{\langle \nabla_{\bbb_{i}}\Phi(\mathbf{b}^t), \bbb_{i} - \bbb_{i}^t\rangle + \frac{2}{\rho_i} \KL(\bbb_{i} || \bbb_{i}^t)\right\};
\end{align*}
if $-\infty < \rho_i < 0$ for buyer $i$, then the Damped Proportional Response (given by \eqref{update::damped}) is equivalent to mirror descent (really ascent as this is a concave function) with a halved step size defined as follows:
\begin{align*}
\bbb_{i}^{t+1} = {\arg \min}_{\bbb_{i} : \sum_j b_{ij} = e_i} \left\{- \langle \nabla_{\bbb_{i}}\Phi(\mathbf{b}^t), \bbb_{i} - \bbb_{i}^t\rangle + \frac{2(\rho_i - 1)}{\rho_i} \KL(\bbb_{i} || \bbb_{i}^t)\right\};
\end{align*}
and if $\rho_i = -\infty$ for buyer $i$, then the Damped Proportional Response (given by \eqref{update::damped}) is equivalent to mirror descent (really ascent as this is a concave function) with a halved step sizem defined as follows:
\begin{align*}
\bbb_{i}^{t+1} = {\arg \min}_{\bbb_{i} : \sum_j b_{ij} = e_i} \{- \langle \nabla_{\bbb_{i}}\Phi(\mathbf{b}^t), \bbb_{i} - \bbb_{i}^t\rangle + 2 \KL(\bbb_{i} || \bbb_{i}^t)\}.
\end{align*}
\end{lemma}
\begin{proof}
By calculation.
\end{proof}

By Lemma~\ref{ineq::saddle::upper::lower} and with a simple calculation one can show that, in Definition~\ref{defn::convex::concave}, if we set $\bbx = \mathbf{b}_{>0}$, $\bby = \mathbf{b}_{<0}$, $d_g(\bbx) = \sum_{i: \rho_i > 0} \frac{1}{\rho_i} \KL(\bbb_{i} || \bbb_{i}')$, and $d_h(\bby) = \sum_{i: \infty < \rho_i < 0} \frac{\rho_i - 1}{\rho_i} \KL(\bbb_{i} || \bbb_{i}') + \sum_{i: \rho_i = -\infty} \KL(\bbb_{i} || \bbb_{i}')$ , then $\Phi$ is $(1, 1)$-convex-concave function.

Furthermore, let $\mathbf{b}^*_{>0}$ and $\mathbf{b}^*_{<0}$ be the market equilibrium of the Fisher market.
Then,
\begin{align*}
\mathbf{b}^*_{>0} \text{ minimizes } \Phi(\cdot, \mathbf{b}^*_{<0}), \hhsp\text{and}\hhsp
\mathbf{b}^*_{<0} \text{ maximizes } \Phi(\mathbf{b}^*_{>0}, \cdot),
\end{align*}
which implies $(\mathbf{b}^*_{>0}, \mathbf{b}^*_{<0})$ is a saddle point of the potential function $\Phi$.  Theorem~\ref{thm:sadd::conv::sublinear} yields the following corollary.
\begin{corollary}
\label{cor:full-range-no-constraints}
The Damped Proportional Response (given by \eqref{update::damped}) converges to an equilibrium with a convergence rate of: 
\begin{align*}
&\sum_{t = 1}^{T} \Bigg[\Phi(\mathbf{b}^{t}_{>0},  \mathbf{b}^{*}_{<0}) - \Phi (\mathbf{b}^{*}_{>0}, \mathbf{b}^{t}_{<0}) \Bigg] \\
 &~~~~~~~~~~~~\leq  \sum_{i: -\infty < \rho_i < 0} \frac{2(\rho_i - 1)}{\rho_i}  \KL(\bbb_{i}^* || \bbb_{i}^0) + \sum_{i: \rho_i > 0} \frac{2}{\rho_i}  \KL(\bbb_{i}^* || \bbb_{i}^0)+ \sum_{i: \rho_i = -\infty} 2 \KL(\bbb_{i}^* || \bbb_{i}^0).
\end{align*}
\end{corollary}

Moreover, if we assume there is no buyer with either a linear utility of a Leontief utility function, then, by Lemma~\ref{ineq::saddle::upper::lower},  $\Psi(\cdot, \cdot)$ is a $(\min_{i:\rho_i > 0} \{1 - \rho_i\}, \min_{i: \rho_i < 0} \{\frac{1}{1 - \rho_i}\},1,1)$-strong Bregman convex-concave function with 
$\bbx = \mathbf{b}_{>0}$, $\bby = \mathbf{b}_{<0}$, $d_g(\bbx) = \sum_{i: \rho_i > 0} \frac{1}{\rho_i} \KL(\bbb_{i} || \bbb_{i}')$ and $d_h(\bby) = \sum_{i: \infty < \rho_i < 0} \frac{\rho_i - 1}{\rho_i} \KL(\bbb_{i} || \bbb_{i}')$ (see Definition~\ref{defn::strong::saddle}). 
Theorem~\ref{conv::saddle::linear} yields the following corollary.
\begin{corollary}
\label{cor:no-extremes}
Suppose there is no buyer with either a linear utility or a Leontief utility. Let 
\begin{align*}
\sigma_{>0} = \min_{i:\rho_i > 0} \{1 - \rho_i\} \hsp\mbox{and}\hsp \sigma_{<0} =  \min_{i:\rho_i < 0} \left\{\frac{1}{1- \rho_i}\right\}.
\end{align*}
Then the Damped Proportional Response (given by \eqref{update::damped}) converges to the equilibrium with a convergence rate of
\begin{align*}
 &\Phi(\mathbf{b}^{T}_{>0},  \mathbf{b}^{*}_{<0}) - \Phi (\mathbf{b}^{*}_{>0}, \mathbf{b}^{T}_{<0}) \\
&\hspace*{0.2in}\leq \Big(1 - \frac{\min\{\sigma_{>0}, \sigma_{<0}\}}{2}\Big)^{T-1} \Bigg[ \sum_{i: \rho_i < 0} (2- \sigma_{<0}) \frac{\rho_i - 1}{\rho_i}  \KL(\bbb_{i}^* || \bbb_{i}^0)  + \sum_{i: \rho_i > 0} (2- \sigma_{>0}) \frac{1}{\rho_i} \KL(\bbb_{i}^* || \bbb_{i}^0)\Bigg]. 
\end{align*} 
\end{corollary}
\emph{The Entire CES Range}\hspace*{0.1in}
Now we consider the Damped Proportional Response with a damping factor of $2$ over the entire CES range,
i.e.\ including Cobb-Douglas utilities. 
Recall that we modify our potential function to include terms for the buyers with $\rhoi = 0$.
However, for this new modified function, 
the buyers with Cobb-Douglas utility functions don't actually perform mirror descent. 
Fortunately, we can make two observations. 

First, the buyers with Cobb-Douglas utility functions converge quickly to the equilibrium independently of everyone else's spending. Second, the buyers whose utility functions are not Cobb-Douglas will still perform the  mirror descent (ascent) procedure. 

So, intuitively, in our analysis, we regard the spending of the buyers with Cobb-Douglas utility functions as a parameter, $\theta$, of $f_{\theta}(\bbx,\bby)$, where $\bbx$ represents the spending of the strictly substitutes buyers and $\bby$ represents the spending of the strictly complementary buyers. Remember, in the 
case with no Cobb-Douglas utilities, the market equilibrium corresponded to a saddle point. Here, similarly, a market equilibrium corresponds to a saddle point of $f_{\mathbf{\theta}^*}(\cdot, \cdot)$, where $\mathbf{\theta}^*$ is the spending at the market equilibrium of those buyers with Cobb-Douglas utility functions. We prove the following two claims.

\begin{enumerate}
\item $\mathbf{\theta}$ converges to $\mathbf{\theta}^*$ quickly;
\item when $\mathbf{\theta}$ tends to  $\mathbf{\theta}^*$, though $\bbx$ and $\bby$ perform the mirror descent based on the gradient of $f_{\mathbf{\theta}}(\bbx,\bby)$ and not of $f_{\mathbf{\theta}^*}(\bbx, \bby)$, $(\bbx, \bby)$ will still converge quickly to $(\bbx^*, \bby^*)$,
the saddle point of $f_{\mathbf{\theta}^*}(\cdot, \cdot)$.
\end{enumerate}

We thereby show that Damped Proportional Response converges to the market equilibrium even when faced with the entire range of CES utilities.
\section{Other Measures of Convergence}\label{sec:discussion}

The potential function $\phi$ appears to be closely related to the Eisenberg-Gale program.
In particular, 
we can show that
in the substitutes domain,
when applying update rule \eqref{eqn:subst-CES-PR-rule},
the Proportional Response update, 
the objective function $\Psi$ for the Eisenberg-Gale program
converges at least as fast as $\Phi$, i.e.\ that 
$\Psi(\mathbf{x}(\mathbf{b}^*_{>0}, \mathbf{b}^*_{=0})) - \Psi(\mathbf{x}(\mathbf{b}_{>0}, \mathbf{b}^*_{=0})) \leq \Phi(\mathbf{b}_{>0}, \mathbf{b}^*_{=0}) -\Phi(\mathbf{b}^*_{>0}, \mathbf{b}^*_{=0})$,
and that in the complementary domain,
when applying update rule \eqref{eqn:comp-CES-PR-rule}, 
the objective function for the dual of the Eisenberg-Gale program
converges at least as fast as $\Phi$.
These claims are shown in Appendix~\ref{sec:Egand-ourPotFn}.

Lemma~\ref{ineq::saddle::upper::lower} allows us to make some observations about the rate of 
convergence of the spending.
For update rule \eqref{eqn:subst-CES-PR-rule}, in the substitutes domain excluding linear utilities,
we can deduce that $\sum_i\KL(\bbb_i || \bbb_i^*) \le \max_i \frac {\rhoi} {1 - \rhoi} [\Phi(\bbb) -{\tiny } \Phi(\bbb^*)]$,
and for update rule \eqref{eqn:comp-CES-PR-rule} in the complementary domain excluding Leontief utilities,
that $\sum_i \KL(\bbb_i || \bbb_i^*) \le \max_i {-\rhoi} [\Phi(\bbb) - \Phi(\bbb^*)]$.
As the equilibrium need not be unique in terms of spending for either linear or Leontief utilities,
this lemma is not going to yield a bound on the convergence rate of the spending in these cases,
as it can be applied to any equilibrium.
Similarly in the combined domain, still excluding linear and Leontief utilities,
we can observe that
\[
\sum_{i: \rhoi > 0}  \frac {1 - \rhoi} {\rhoi} \KL(\bbb_i || \bbb_i^*) +
\sum_{i: \rhoi < 0}\frac {-1}{\rhoi} \KL(\bbb_i || \bbb_i^*) 
\le [\Phi(\bbb_{>0},\bbb_{<0}^*) - \Phi(\bbb^*)] + [\Phi(\bbb^*) - \Phi(\bbb_{>0}^*,\bbb_{<0})].
\]
In addition, since $\KL(\bbp||\bbp^*) \le \KL(\bbb||\bbb^*)$ we can immediately obtain analogous
bounds on the KL divergence of the prices.
Furthermore, for the substitutes domain, including linear utilities,
Lemma~\ref{ineq::saddle::upper::lower} also implies that $\KL(\bbp||\bbp^*) \le  [\Phi(\bbb) - \Phi(\bbb^*)]$.

\section{Open Problems}

It would be interesting to understand what happens when asynchronous updating is allowed,
i.e.\ each buyers makes its updates at times independent of the other updates,
though with some notion of bounded asynchrony.
However, at present, to the best of our knowledge,
there are no techniques for analyzing asynchronous forms of mirror descent,
which makes this a substantial challenge.
Indeed, whether even convergence occurs, let alone fast convergence, is not clear-cut.

Another open problem concerns Proportional Response in more general Arrow-Debreu markets.
Wu and Zhang~\cite{WZ2007} gave an analysis for the substitute domain in which each agent $i$ is endowed one unit of good $i$;
they imposed two restrictions on the utility functions of the agents:
first, each agent has the same positive $\rho$-parameter, and second,
the coefficients in the agents' utility functions satisfy: $a_{ij} > 0$ if and only if $a_{ji} > 0$ (see their Theorem 4.1).
In contrast, market equilibria are known to exist in more general situations,
and admit a convex formulation for the case $\min_i \rho_i > -1$~\cite{CMPV2005}.
It would be interesting to understand the behavior of Proportional Response (or its variants)
if the restrictions imposed by Wu and Zhang are removed.

\section*{Acknowledgment}
We are grateful to the anonymous referees for their insightful comments which helped improve the presentation.

% Bibliography
\bibliographystyle{plain}
\bibliography{bib}

\newpage
\appendix
\section{Proportional Response Including Cobb-Douglas Utilities}
\label{sec::entire_range}

\subsection{The Potential Function and its Properties}\label{construction::cobb}
 The new potential function, $\Phi(\mathbf{b})$, which includes buyers with Cobb-Douglas utilities, is defined as follows:
\begin{align*}
&p_j(\mathbf{b}) = \sum_{i} b_{ij}, \\
&\Phi(\mathbf{b}) = - \sum_{i: \rho_i \neq \{0, -\infty\}} \frac{1}{\rho_i} \sum_{j} b_{ij} \log \frac{a_{ij} b_{ij}^{\rho_i - 1}}{[p_j(\mathbf{b})]^{\rho_i}} - \sum_{i: \rho_i = -\infty} \sum_j b_{ij} \log \frac{b_{ij}}{c_{ij}p_j(\mathbf{b})} + \sum_{i: \rho_i = 0} \sum_j b_{ij} \log p_j(\mathbf{b}).
\end{align*}
Note that for those $i$ for which $\rho_i \neq \{0, -\infty\}$,
\begin{align*}
\nabla_{b_{ij}} \Phi(\mathbf{b}) &= - \frac{1}{\rho_i} \log a_{ij} - \frac{\rho_i - 1}{\rho_i}(\log b_{ij} + 1) + \log p_j(\mathbf{b}) + \sum_{h}  b_{hj} \frac{1}{p_j(\mathbf{b})} \\
&= \frac{1}{\rho_i}\Big(1 - \log \frac{a_{ij} b_{ij}^{\rho_i - 1}}{[p_j(\mathbf{b})]^{\rho_i}}\Big);
\end{align*}
for those $i$ for which $\rho_i = -\infty$,
\begin{align*}
\nabla_{b_{ij}} \Phi(\mathbf{b}) = - \log \frac{ b_{ij}}{c_{ij} p_j(\mathbf{b})};
\end{align*}
and for those $i$ for which $\rho_i = 0$,
\begin{align*}
\nabla_{b_{ij}} \Phi(\mathbf{b}) = 1 + \log p_j(\mathbf{b}).
\end{align*}

Then, we can deduce that
\begin{align*}
\Phi(\mathbf{b}) - \Phi(\mathbf{b}') - \langle \nabla \Phi(\mathbf{b}'), \mathbf{b} - \mathbf{b}' \rangle &= -\sum_{i: \rho_i \neq \{0, -\infty\}} \frac{\rho_i - 1}{\rho_i} \KL(\bbb_{i} || \bbb_{i}') - \sum_{i:\rho_i = -\infty} \KL(\bbb_{i} || \bbb_{i}')  + \KL(\bbp || \bbp') \\
&= \sum_{i: \rho_i \neq \{0, -\infty\}} \frac{1}{\rho_i} \KL(\bbb_{i} || \bbb_{i}')  + \sum_{i: \rho_i = 0} \KL(\bbb_{i} || \bbb_{i}')  \\
&\hspace*{1in}- \Big( \sum_i \KL(\bbb_{i} || \bbb_{i}') - \KL(\bbp || \bbp')\Big).
\end{align*}

Since $\sum_i \KL(\bbb_{i} || \bbb_{i}') \geq \KL(\bbp || \bbp')$,
\begin{align*} \label{ineq::saddle::upper::lower::2}
&\sum_{i: \rho_i \neq \{0, -\infty\}} \frac{1 - \rho_i}{\rho_i} \KL(\bbb_{i} || \bbb_{i}') - \sum_{i: \rho_i = -\infty} \KL(\bbb_{i} || \bbb_{i}') \\
&\hspace*{1.4in}\leq \Phi(\mathbf{b}) - \Phi(\mathbf{b}') - \langle \nabla \Phi(\mathbf{b}'), \mathbf{b} - \mathbf{b}' \rangle \\
&\hspace*{1.4in}\hspace*{0.2in}\leq \sum_{i: \rho_i \neq \{0, -\infty\}} \frac{1}{\rho_i} \KL(\bbb_{i} || \bbb_{i}') + \sum_{i: \rho_i = 0} \KL(\bbb_{i} || \bbb_{i}') . \numberthis
\end{align*}

\subsection{Proportional Response in the Substitutes Domain with Cobb-Douglas Utility Functions}

As in the analysis of Proportional Response in the substitutes domain without Cobb-Douglas utility functions, the following lemma states the equivalence between mirror descent and Proportional Response in the substitutes domain; it follows readily from the definition of $\bbb^{t+1}$ for Proportional Response (given by \eqref{eqn:subst-CES-PR-rule}).
\begin{lemma} \label{PR::EQ::SUB::COBB}
For buyers with strict CES substitutes utilities ($\rho_i > 0$), 
the Proportional Response update is the same as the mirror descent update, given by:
\begin{align*}
\bbb_{i}^{t+1} = {\arg \min}_{\bbb_i : \sum_j b_{ij} = e_i} \{\langle \nabla_{\bbb_i}\Phi(\mathbf{b}^t), \bbb_i - \bbb_i^t\rangle + \frac{1}{\rho_i} \KL(\bbb_i || \bbb_i^t)\}.
\end{align*}
\end{lemma}
With a simple calculation, one can show the following properties:
\begin{enumerate}
\item $\bbb^t_{=0}$ will be equal to $\bbb^*_{=0}$ for $t > 0$, which, for $\rho_i = 0$ and $t > 0$,
 implies that:
\begin{align*}\label{eq::zero::sub::cobb}
\KL(\bbb_i^* || \bbb_i^t) = 0;\numberthis
\end{align*} 
\item  $\mathbf{b}_{>0}$ is the spending at the market equilibrium if and only if $\mathbf{b}_{>0}$ is the minimum point of $\Phi(\cdot, \bbb^*_{=0})$. 
\end{enumerate}
We will show the following result.
\begin{theorem} \label{thm::pure::sub::cobb}
\begin{align*}
\Phi (\mathbf{b}^{T}_{>0}, \mathbf{b}^{*}_{=0}) - \Phi (\mathbf{b}^{*}_{>0}, \mathbf{b}^{*}_{=0}) \leq \frac{1}{T} \Bigg(\sum_{i: \rho_i > 0} \frac{1}{\rho_i}  \KL(\bbb_{i}^* || \bbb_{i}^0) + \sum_{i: \rho_i = 0} \KL(\bbb^*_i || \bbb^0_i)\Bigg).
\end{align*}
\end{theorem}
We first show the following lemma.
\begin{lemma}\label{lem::decreasing::function::value}
For $t > 0$,
$\Phi(\bbb^{t+1}_{>0}, \bbb^*_{=0}) \leq \Phi(\bbb^{t}_{>0}, \bbb^*_{=0})$.
\end{lemma}
\begin{proof}
By Lemma~\ref{PR::EQ::SUB::COBB}, we know that for those $i$ for which $\rho_i > 0$,
\begin{align*}
\bbb_{i}^{t+1} = {\arg \min}_{\bbb_{i} : \sum_j b_{ij} = e_i} \left\{\langle \nabla_{\bbb_{i}}\Phi(\mathbf{b}^t), \bbb_{i} - \bbb_{i}^t\rangle + \frac{1}{\rho_i} \KL(\bbb_{i} || \bbb_{i}^t)\right\}.
\end{align*}
Therefore, by Lemma~\ref{ineq::basic::bregman} with $x^+ = \mathbf{b}^{t+1}_{>0}$, $x = \mathbf{b}^{t}_{>0}$, and $y = \mathbf{b}^t_{>0}$,
\begin{align*}
&\langle \nabla_{\mathbf{b}_{>0}} \Phi(\mathbf{b}^t_{>0}, \mathbf{b}^t_{=0}), \mathbf{b}^{t+1}_{>0} - \mathbf{b}^t_{>0} \rangle + \sum_{i: \rho_i > 0} \frac{1}{\rho_i} \KL(\bbb_{i}^{t+1} || \bbb_{i}^t) \\
&\hspace*{0.2in}\leq \langle \nabla_{\mathbf{b}_{>0}} \Phi(\mathbf{b}^t_{>0}, \mathbf{b}^t_{=0}), \mathbf{b}^{t}_{>0} - \mathbf{b}^t_{>0} \rangle + \sum_{i: \rho_i > 0} \frac{1}{\rho_i}  \KL(\bbb_{i}^t || \bbb_{i}^t)  - \sum_{i: \rho_i > 0} \frac{1}{\rho_i}  \KL(\bbb_{i}^t || \bbb_{i}^{t+1}) \leq 0.
\end{align*}
Also, we know that for $t > 0$, $\bbb^t_{=0} = \bbb^*_{=0}$. Then, for $t > 0$:
\begin{align*}
\langle \nabla \Phi(\mathbf{b}^t_{>0}, \mathbf{b}^*_{=0}), (\mathbf{b}^{t+1}_{>0}, \mathbf{b}^*_{=0})- (\mathbf{b}^t_{>0}, \mathbf{b}^*_{=0}) \rangle + \sum_{i: \rho_i > 0} \frac{1}{\rho_i} \KL(\bbb_{i}^{t+1} || \bbb_{i}^t)  \leq 0.
\end{align*}
Applying \eqref{ineq::saddle::upper::lower::2} 
with $\bbb = (\bbb^{t+1}_{>0},\bbb^*_{=0})$ and
$\bbb' = (\bbb^t_{>0},\bbb^*_{=0})$
yields:
\begin{align*}
\Phi(\mathbf{b}^{t+1}_{>0}, \mathbf{b}^*_{=0}) - \Phi(\mathbf{b}^t_{>0}, \mathbf{b}^*_{=0}) \leq 0,
\end{align*}
which gives the result.
\end{proof}
\begin{proof}[Proof of Theorem~\ref{thm::pure::sub::cobb}]
 By Lemma~\ref{PR::EQ::SUB::COBB}, 
\begin{align*}
\bbb_{i}^{t+1} = {\arg \min}_{\bbb_{i} : \sum_j b_{ij} = e_i} \left\{\langle \nabla_{\bbb_{i}}\Phi(\mathbf{b}^t), \bbb_{i} - \bbb_{i}^t\rangle + \frac{1}{\rho_i} \KL(\bbb_{i} || \bbb_{i}^t)\right\}.
\end{align*}
Then, by Lemma~\ref{ineq::basic::bregman} with $x^+ = \mathbf{b}^{t+1}_{>0}$, $x = \mathbf{b}^{*}_{>0}$, 
and $y = \mathbf{b}^t_{>0}$,
\begin{align*}
&\langle \nabla_{\mathbf{b}_{>0}} \Phi(\mathbf{b}^t_{>0}, \mathbf{b}^t_{=0}), \mathbf{b}^{t+1}_{>0} - \mathbf{b}^t_{>0} \rangle + \sum_{i: \rho_i > 0} \frac{1}{\rho_i} \KL(\bbb_{i}^{t+1} || \bbb_{i}^t) \\
&\hspace*{0.2in}\leq \langle \nabla_{\mathbf{b}_{>0}} \Phi(\mathbf{b}^t_{>0}, \mathbf{b}^t_{=0}), \mathbf{b}^{*}_{>0} - \mathbf{b}^t_{>0} \rangle + \sum_{i: \rho_i > 0} \frac{1}{\rho_i}  \KL(\bbb_{i}^* || \bbb_{i}^t)  - \sum_{i: \rho_i > 0} \frac{1}{\rho_i}  \KL(\bbb_{i}^* || \bbb_{i}^{t+1}).
\end{align*}

This is equivalent to:
\begin{align*} \label{ineq::sub::cobb::pos::triangle}
&\underbrace{\langle \nabla \Phi(\mathbf{b}^t_{>0}, \mathbf{b}^t_{=0}), (\mathbf{b}^{t+1}_{>0}, \mathbf{b}^{*}_{=0}) - (\mathbf{b}^t_{>0},\mathbf{b}^{t}_{=0}) \rangle + \sum_{i: \rho_i > 0} \frac{1}{\rho_i} \KL(\bbb_{i}^{t+1} || \bbb_{i}^t)}_{\texttt{LHS}} \\
&~~~~~~~~~\leq \underbrace{\langle \nabla \Phi(\mathbf{b}^t_{>0}, \mathbf{b}^t_{=0}), (\mathbf{b}^{*}_{>0}, \mathbf{b}^{*}_{=0}) - (\mathbf{b}^t_{>0},\mathbf{b}^{t}_{=0}) \rangle + \sum_{i: \rho_i > 0} \frac{1}{\rho_i}  \KL(\bbb_{i}^* || \bbb_{i}^t)  - \sum_{i: \rho_i > 0} \frac{1}{\rho_i}  \KL(\bbb_{i}^* || \bbb_{i}^{t+1})}_{\texttt{RHS}}. \numberthis
\end{align*} 

From the second inequality in \eqref{ineq::saddle::upper::lower::2} with
$\bbb=(\bbb^{t+1}_{>0},\bbb^*_{=0})$ and $\bbb'=(\bbb^t_{>0},\bbb^t_{=0})$,
the LHS term is lower bounded by:
\begin{align*} \label{ineq::sub::cobb::pos::lhs}
\Phi (\mathbf{b}^{t+1}_{>0}, \mathbf{b}^{*}_{=0}) - \Phi(\mathbf{b}^t_{>0}, \mathbf{b}^t_{=0}) - \sum_{i: \rho_i > 0} \frac{1}{\rho_i} \KL(\bbb^{t+1}_i || \bbb^t_i) - \sum_{i: \rho_i = 0} \KL(\bbb^*_i || \bbb^t_i) +  \sum_{i: \rho_i > 0} \frac{1}{\rho_i}\KL(\bbb_{i}^{t+1} || \bbb_{i}^t)  \numberthis
\end{align*}
and
from the first inequality with $\bbb=(\bbb^*_{>0},\bbb^*_{=0})$ and
$\bbb'=(\bbb^t_{>0},\bbb^t_{=0})$,
the RHS term is upper bounded by:
\begin{align*} \label{ineq::sub::cobb::pos::rhs}
& \Phi (\mathbf{b}^{*}_{>0}, \mathbf{b}^{*}_{=0}) - \Phi(\mathbf{b}^t_{>0}, \mathbf{b}^t_{=0}) - \underbrace{\sum_{i: \rho_i > 0} \frac{1 - \rho_i}{\rho_i} \KL(\bbb_{i}^* || \bbb_{i}^t)}_{B} + \sum_{i: \rho_i > 0} \frac{1}{\rho_i}  \KL(\bbb_{i}^* || \bbb_{i}^t)  - \sum_{i: \rho_i > 0} \frac{1}{\rho_i}  \KL(\bbb_{i}^* || \bbb_{i}^{t+1}). \numberthis
\end{align*}

As $\texttt{LHS} \leq \texttt{RHS}$, and as $B$ is positive, we have:
\begin{align*}\label{ineq::sub::cobb::pos::upper}
&\Phi (\mathbf{b}^{t+1}_{>0}, \mathbf{b}^{*}_{=0})  - \sum_{i: \rho_i = 0} \KL(\bbb^*_i || \bbb^t_i)  
\leq \Phi (\mathbf{b}^{*}_{>0}, \mathbf{b}^{*}_{=0})  + \sum_{i: \rho_i > 0} \frac{1}{\rho_i}  \KL(\bbb_{i}^* || \bbb_{i}^t)  - \sum_{i: \rho_i > 0} \frac{1}{\rho_i}  \KL(\bbb_{i}^* || \bbb_{i}^{t+1}). \numberthis
\end{align*}

Summing over $t$ gives:
\begin{align*}
\sum_{t = 0}^{T-1} \Big(\Phi (\mathbf{b}^{t+1}_{>0}, \mathbf{b}^{*}_{=0}) - \Phi (\mathbf{b}^{*}_{>0}, \mathbf{b}^{*}_{=0})\Big) &\leq \sum_{i: \rho_i > 0} \frac{1}{\rho_i}  \KL(\bbb_{i}^* || \bbb_{i}^0) + \sum_{t = 0}^{T-1} \sum_{i: \rho_i = 0} \KL(\bbb^*_i || \bbb^t_i)  \\
&=  \sum_{i: \rho_i > 0} \frac{1}{\rho_i}  \KL(\bbb_{i}^* || \bbb_{i}^0) + \sum_{i: \rho_i = 0} \KL(\bbb^*_i || \bbb^0_i).
\end{align*}
The second equality holds because of \eqref{eq::zero::sub::cobb}.

By Lemma~\ref{lem::decreasing::function::value},
\begin{align*}
\Phi (\mathbf{b}^{T}_{>0}, \mathbf{b}^{*}_{=0}) - \Phi (\mathbf{b}^{*}_{>0}, \mathbf{b}^{*}_{=0}) \leq \frac{1}{T} \Bigg(\sum_{i: \rho_i > 0} \frac{1}{\rho_i}  \KL(\bbb_{i}^* || \bbb_{i}^0) + \sum_{i: \rho_i = 0} \KL(\bbb^*_i || \bbb^0_i)\Bigg).
\end{align*}
\end{proof}
\begin{theorem}
Suppose there is no buyer with a linear utility function. Let $\sigma = \Big(\min_{i : \rho_i > 0}\left\{\frac{1}{\rho_i}\right\}\Big)$.
Then,
\begin{align*}
\Phi (\mathbf{b}^{T}_{>0}, \mathbf{b}^{*}_{=0}) - \Phi (\mathbf{b}^{*}_{>0}, \mathbf{b}^{*}_{=0}) \leq \frac{\sigma - 1}{\sigma^{T} - 1} \Bigg(\sum_{i: \rho_i > 0} \KL(\bbb_{i}^* || \bbb_{i}^0) + \sum_{i: \rho_i = 0} \KL(\bbb^*_i || \bbb^0_i) \Bigg).
\end{align*}
\end{theorem}
\begin{proof}
If there is no buyer with a linear utility function, then we do not drop B in \eqref{ineq::sub::cobb::pos::rhs}. So, instead of \eqref{ineq::sub::cobb::pos::upper}, we have:
\begin{align*}
&\Phi (\mathbf{b}^{t+1}_{>0}, \mathbf{b}^{*}_{=0})  - \sum_{i: \rho_i = 0} \KL(\bbb^*_i || \bbb^t_i)  \\
&\hspace*{0.2in}~~~~~~~~~\leq \Phi (\mathbf{b}^{*}_{>0}, \mathbf{b}^{*}_{=0})  + \sum_{i: \rho_i > 0} \KL(\bbb_{i}^* || \bbb_{i}^t)  - \sum_{i: \rho_i > 0} \frac{1}{\rho_i}  \KL(\bbb_{i}^* || \bbb_{i}^{t+1}). \numberthis
\end{align*}

Multiplying both sides by $\Big(\min_{i : \rho_i > 0}\left\{\frac{1}{\rho_i}\right\}\Big)^t$ and summing over all $t$ yields:
\begin{align*}
\sum_{t = 0}^{T-1} \Big(\min_{i : \rho_i > 0}\left\{\frac{1}{\rho_i}\right\}\Big)^t \Big(\Phi (\mathbf{b}^{t+1}_{>0}, \mathbf{b}^{*}_{=0}) - \Phi (\mathbf{b}^{*}_{>0}, \mathbf{b}^{*}_{=0})\Big) &\leq \sum_{i: \rho_i > 0} \KL(\bbb_{i}^* || \bbb_{i}^0) + \sum_{i: \rho_i = 0} \KL(\bbb^*_i || \bbb^0_i).
\end{align*}
Recall that $\sigma = \Big(\min_{i : \rho_i > 0}\left\{\frac{1}{\rho_i}\right\}\Big)$. By Lemma~\ref{lem::decreasing::function::value}, 
\begin{align*}
\Phi (\mathbf{b}^{T}_{>0}, \mathbf{b}^{*}_{=0}) - \Phi (\mathbf{b}^{*}_{>0}, \mathbf{b}^{*}_{=0}) \leq \frac{\sigma - 1}{\sigma^{T} - 1} \Bigg(\sum_{i: \rho_i > 0} \KL(\bbb_{i}^* || \bbb_{i}^0) + \sum_{i: \rho_i = 0} \KL(\bbb^*_i || \bbb^0_i) \Bigg).
\end{align*}
\end{proof}

\subsection{Proportional Response in the Complementary Domain with Cobb-Douglas Utility Functions}

The argument in this case is quite similar to the one in the previous subsection.

First, the following lemma shows the equivalence between mirror descent and  Proportional Response (given by \eqref{eqn:comp-CES-PR-rule}) in the complementary domain.
\begin{lemma}\label{lem::updating::comple::cobb}
For those complementary buyers such that $-\infty< \rho_i < 0$, Proportional Response, 
which is the best response in this domain, is equivalent to the mirror descent update, given by:
\begin{align*}
\bbb_{i}^{t+1} = {\arg \min}_{\bbb_i : \sum_j b_{ij} = e_i} \left\{- \langle \nabla_{b_i}\Phi(\mathbf{b}^t), \bbb_i - \bbb_i^t\rangle + \frac{\rho_i - 1}{\rho_i} \KL(\bbb_i || \bbb_i^t)\right\};
\end{align*}
and for those buyers with Leontief utility functions, Proportional Response, 
which is also the best response in this domain, is equivalent to the mirror descent update, given by:
\begin{align*}
\bbb_{i}^{t+1} = {\arg \min}_{\bbb_i : \sum_j b_{ij} = e_i} \left\{- \langle \nabla_{b_i}\Phi(\mathbf{b}^t), \bbb_i - \bbb_i^t\rangle +  \KL(\bbb_i || \bbb_i^t)\right\}.
\end{align*}
\end{lemma}
We also have the following properties:
\begin{enumerate}
\item $\bbb^t_{=0}$ will be equal to $\bbb^*_{=0}$ for $t > 0$, which, for $\rho_i = 0$ and $t > 0$,
implies that:
\begin{align*}\label{eq::zero::compl::cobb}
\KL(\bbb_i^* || \bbb_i^t) = 0;\numberthis
\end{align*} 
\item  $\mathbf{b}_{<0}$ is the spending at the market equilibrium if and only if $\mathbf{b}_{<0}$ is the 
maximum point of $\Phi(\cdot, \bbb^*_{=0})$. 
\end{enumerate}
We show the following result.
\begin{theorem}\label{thm::conver::comple::cobb::T}
\begin{align*}
\Phi(\mathbf{b}^{*}_{=0}, \mathbf{b}^{*}_{<0})  - \Phi(\mathbf{b}^{*}_{=0}, \mathbf{b}^{T}_{<0}) \leq \frac{1}{T}\Big(\sum_{i: \rho_i = 0} \KL(\bbb_{i}^* || \bbb_{i}^0) + \sum_{i: -\infty < \rho_i < 0} \frac{\rho_i - 1}{\rho_i}  \KL(\bbb_{i}^* || \bbb_{i}^0)  + \sum_{i : \rho_i = -\infty}  \KL(\bbb_{i}^* || \bbb_{i}^0)\Big).
\end{align*}
\end{theorem}
We first show the following lemma.
\begin{lemma}\label{lem::decreasing::compl::cobb}
For $t > 0$,
$\Phi( \mathbf{b}^*_{=0}, \mathbf{b}^t_{<0}) \leq \Phi(\mathbf{b}^*_{=0}, \mathbf{b}^{t+1}_{<0})$.
\end{lemma}
\begin{proof}
By Lemma~\ref{lem::updating::comple::cobb},  we know that for those $i$ for which $-\infty < \rho_i < 0$,
\begin{align*}
\bbb_{i}^{t+1} = {\arg \min}_{\bbb_{i} : \sum_j b_{ij} = e_i} \left\{- \langle \nabla_{\bbb_{i}}\Phi(\mathbf{b}^t), \bbb_{i} - \bbb_{i}^t\rangle + \frac{\rho_i - 1}{\rho_i} \KL(\bbb_{i} || \bbb_{i}^t)\right\},
\end{align*}
and for those $i$ for which $\rho_i = -\infty$,
\begin{align*}
\bbb_{i}^{t+1} = {\arg \min}_{\bbb_{i} : \sum_j b_{ij} = e_i} \{- \langle \nabla_{\bbb_{i}}\Phi(\mathbf{b}^t), \bbb_{i} - \bbb_{i}^t\rangle + \KL(\bbb_{i} || \bbb_{i}^t)\}.
\end{align*}

Therefore, by Lemma~\ref{ineq::basic::bregman}, 
with $x^+=\bbb^{t+1}_{<0}$, $x = \bbb^t_{<0}$, and $y =\bbb^t_{<0}$,
\begin{align*}
&- \langle \nabla_{\mathbf{b}_{<0}} \Phi(\mathbf{b}^t_{=0}, \mathbf{b}^t_{<0}), \mathbf{b}^{t+1}_{<0} - \mathbf{b}^t_{<0} \rangle +  \sum_{i: -\infty < \rho_i <0}\frac{\rho_i - 1}{\rho_i} \KL(\bbb_{i}^{t+1} || \bbb_{i}^t) + \sum_{i : \rho_i = -\infty}  \KL(\bbb_{i}^{t+1} || \bbb_{i}^t) \\
&\hspace*{0.2in}\leq -\langle \nabla_{\mathbf{b}_{<0}} \Phi(\mathbf{b}^t_{=0}, \mathbf{b}^t_{<0}), \mathbf{b}^{t}_{<0} - \mathbf{b}^t_{<0} \rangle + \sum_{i: -\infty < \rho_i < 0} \frac{\rho_i - 1}{\rho_i}  \KL(\bbb_{i}^t || \bbb_{i}^t)  - \sum_{i: -\infty < \rho_i < 0} \frac{\rho_i - 1}{\rho_i}  \KL(\bbb_{i}^t || \bbb_{i}^{t+1})\\
&\hspace*{0.2in}\hspace*{0.2in}~~~~~~~~~~~~~~~~~~+ \sum_{i : \rho_i = -\infty} \KL(\bbb_{i}^t || \bbb_{i}^t) -\sum_{i : \rho_i = -\infty}  \KL(\bbb_{i}^t || \bbb_{i}^{t+1}) \\
&\hspace*{0.2in}\leq 0.
\end{align*}

We know that for $t > 0$, $\bbb^{t}_{=0} = \bbb^*_{=0}$. Therefore, for $t > 0$,
\begin{align*}
&- \langle \nabla \Phi(\mathbf{b}^*_{=0}, \mathbf{b}^t_{<0}), (\mathbf{b}^*_{=0}, \mathbf{b}^{t+1}_{<0}) - (\mathbf{b}^*_{=0}, \mathbf{b}^t_{<0}) \rangle +  \sum_{i: -\infty < \rho_i <0}\frac{\rho_i - 1}{\rho_i} \KL(\bbb_{i}^{t+1} || \bbb_{i}^t) + \sum_{i : \rho_i = -\infty}  \KL(\bbb_{i}^{t+1} || \bbb_{i}^t)  \\
&\hspace*{1in}\hspace*{1in}\leq 0.
\end{align*}

Using ~\eqref{ineq::saddle::upper::lower::2} yields:
\begin{align*}
\Phi( \mathbf{b}^*_{=0}, \mathbf{b}^t_{<0}) - \Phi(\mathbf{b}^*_{=0}, \mathbf{b}^{t+1}_{<0}) \leq 0.
\end{align*}
\end{proof}
\begin{proof}[Proof of Theorem~\ref{thm::conver::comple::cobb::T}]
First, by Lemma \ref{lem::updating::comple::cobb}, we know that for those $i$ for which $-\infty < \rho_i < 0$,
\begin{align*}
\bbb_{i}^{t+1} = {\arg \min}_{\bbb_{i} : \sum_j b_{ij} = e_i} \left\{- \langle \nabla_{\bbb_{i}}\Phi(\mathbf{b}^t), \bbb_{i} - \bbb_{i}^t\rangle + \frac{\rho_i - 1}{\rho_i} \KL(\bbb_{i} || \bbb_{i}^t)\right\},
\end{align*}
and for those $i$ for which $\rho_i = -\infty$,
\begin{align*}
\bbb_{i}^{t+1} = {\arg \min}_{\bbb_{i} : \sum_j b_{ij} = e_i} \{- \langle \nabla_{\bbb_{i}}\Phi(\mathbf{b}^t), \bbb_{i} - \bbb_{i}^t\rangle + \KL(\bbb_{i} || \bbb_{i}^t)\},
\end{align*}

Therefore, by Lemma~\ref{ineq::basic::bregman}, 
with $x^+=\bbb^{t+1}_{<0}$, $x = \bbb^{*}_{<0}$, and $y = \bbb^t_{<0}$,
\begin{align*}
&- \langle \nabla_{\mathbf{b}_{<0}} \Phi( \mathbf{b}^t_{=0}, \mathbf{b}^t_{<0}), \mathbf{b}^{t+1}_{<0} - \mathbf{b}^t_{<0} \rangle +  \sum_{i: -\infty < \rho_i <0}\frac{\rho_i - 1}{\rho_i} \KL(\bbb_{i}^{t+1} || \bbb_{i}^t) + \sum_{i : \rho_i = -\infty}  \KL(\bbb_{i}^{t+1} || \bbb_{i}^t) \\
&\hspace*{0.2in}\leq -\langle \nabla_{\mathbf{b}_{<0}} \Phi(\mathbf{b}^t_{=0}, \mathbf{b}^t_{<0}), \mathbf{b}^{*}_{<0} - \mathbf{b}^t_{<0} \rangle + \sum_{i: -\infty < \rho_i < 0} \frac{\rho_i - 1}{\rho_i}  \KL(\bbb_{i}^* || \bbb_{i}^t)  - \sum_{i: -\infty < \rho_i < 0} \frac{\rho_i - 1}{\rho_i}  \KL(\bbb_{i}^* || \bbb_{i}^{t+1})\\
&\hspace*{0.2in}\hspace*{0.2in}~~~~~~~~~~~~~~~~~~+ \sum_{i : \rho_i = -\infty} \KL(\bbb_{i}^* || \bbb_{i}^t) -\sum_{i : \rho_i = -\infty} \KL(\bbb_{i}^* || \bbb_{i}^{t+1}).
\end{align*}

This is equivalent to:
\begin{align*} \label{ineq::comp::neg::triangle}
&- \langle \nabla \Phi(\mathbf{b}^t_{=0}, \mathbf{b}^t_{<0}), (\mathbf{b}^{*}_{=0}, \mathbf{b}^{t+1}_{<0}) - (\mathbf{b}^{t}_{=0}, \mathbf{b}^{t}_{<0})  \rangle \\
&\underbrace{+  \sum_{i: -\infty <\rho_i <0}\frac{\rho_i - 1}{\rho_i} \KL(\bbb_{i}^{t+1} || \bbb_{i}^t)+ \sum_{i : \rho_i = -\infty} \KL(\bbb_{i}^{t+1} || \bbb_{i}^t)~~~~~~~~~}_{\texttt{LHS}} \\
&\hspace*{0.2in}~~~~~~~~~\leq -\langle \nabla \Phi(\mathbf{b}^t_{=0}, \mathbf{b}^t_{<0}), ( \mathbf{b}^{*}_{=0}, \mathbf{b}^{*}_{<0}) - (\mathbf{b}^{t}_{=0}, \mathbf{b}^{t}_{<0})  \rangle \\
&\hspace*{0.2in}\hspace*{0.2in}~~~~~~~~~~~~~~~~~~+ \sum_{i: -\infty < \rho_i < 0} \frac{\rho_i - 1}{\rho_i}  \KL(\bbb_{i}^* || \bbb_{i}^t)  - \sum_{i:  -\infty <\rho_i < 0} \frac{\rho_i - 1}{\rho_i}  \KL(\bbb_{i}^* || \bbb_{i}^{t+1})\\
&\hspace*{0.2in}~~~~~~~~~~~~\hspace*{0.1in}\underbrace{\hspace*{0.1in}~~~~~~+ \sum_{i : \rho_i = -\infty}  \KL(\bbb_{i}^* || \bbb_{i}^t) -\sum_{i : \rho_i = -\infty}  \KL(\bbb_{i}^* || \bbb_{i}^{t+1}).\hspace*{0.2in}\hspace*{0.2in}\hspace*{0.2in}~~~~~~~~~~~~~~~~~~~~~~~~~~~~}_{\texttt{RHS}} \numberthis
\end{align*}

From the first inequality in \eqref{ineq::saddle::upper::lower::2}, the LHS term is lower bounded by:
\begin{align*} \label{ineq::comp::neg::lhs}
& - \Phi( \mathbf{b}^{*}_{=0}, \mathbf{b}^{t+1}_{<0}) + \Phi(\mathbf{b}^{t}_{=0}, \mathbf{b}^{t}_{<0}) - \sum_{i: -\infty < \rho_i < 0} \frac{\rho_i - 1}{\rho_i} \KL(\bbb_{i}^{t+1} || \bbb_{i}^t)  - \sum_{i: \rho_i = -\infty}\KL(\bbb_{i}^{t+1} || \bbb_{i}^t) \\
& \hspace*{0.2in}~~~~~~~~~+  \sum_{i: -\infty < \rho_i <0}\frac{\rho_i - 1}{\rho_i} \KL(\bbb_{i}^{t+1} || \bbb_{i}^t) +  \sum_{i: \rho_i = -\infty}  \KL(\bbb_{i}^{t+1} || \bbb_{i}^t) \numberthis
\end{align*}
and from the second inequality, the RHS term is upper bounded by:
\begin{align*} \label{ineq::comp::neg::rhs}
 & -\Phi(\mathbf{b}^{*}_{=0}, \mathbf{b}^{*}_{<0}) + \Phi(\mathbf{b}^{t}_{=0}, \mathbf{b}^{t}_{<0})  + \underbrace{\sum_{i: -\infty < \rho_i < 0} \frac{1}{\rho_i} \KL(\bbb_{i}^* || \bbb_{i}^t)}_{D} + \sum_{i: \rho_i = 0} \KL(\bbb_{i}^* || \bbb_{i}^t) \\
&\hspace*{0.2in}~~~~~~~~~+ \sum_{i: -\infty < \rho_i < 0} \frac{\rho_i - 1}{\rho_i}  \KL(\bbb_{i}^* || \bbb_{i}^t)  - \sum_{i: -\infty < \rho_i < 0} \frac{\rho_i - 1}{\rho_i}  \KL(\bbb_{i}^* || \bbb_{i}^{t+1})\\
&\hspace*{0.2in}~~~~~~~~~+ \sum_{i : \rho_i = -\infty}  \KL(\bbb_{i}^* || \bbb_{i}^t) -\sum_{i : \rho_i = -\infty}  \KL(\bbb_{i}^* || \bbb_{i}^{t+1}). \numberthis
\end{align*}

As $\texttt{LHS} \leq \texttt{RHS}$, and as $D$ is negative, we have:
\begin{align*}\label{ineq::comp::neg::upper}
- \Phi(\mathbf{b}^{*}_{=0}, \mathbf{b}^{t+1}_{<0}) &~~~~~~~~~\leq  -\Phi(\mathbf{b}^{*}_{=0}, \mathbf{b}^{*}_{<0}) + \sum_{i: \rho_i = 0} \KL(\bbb_{i}^* || \bbb_{i}^t) \\
&~~~~~~~~~~~~+ \sum_{i: -\infty < \rho_i < 0} \frac{\rho_i - 1}{\rho_i}  \KL(\bbb_{i}^* || \bbb_{i}^t)  - \sum_{i: -\infty < \rho_i < 0} \frac{\rho_i - 1}{\rho_i}  \KL(\bbb_{i}^* || \bbb_{i}^{t+1}) \\
&~~~~~~~~~~~~+ \sum_{i : \rho_i = -\infty}  \KL(\bbb_{i}^* || \bbb_{i}^t) -\sum_{i : \rho_i = -\infty} \KL(\bbb_{i}^* || \bbb_{i}^{t+1}).\numberthis
\end{align*}

Summing over all $t$ yields:
\begin{align*}
\sum_{t = 0}^{T-1}  \Big(\Phi(\mathbf{b}^{*}_{=0}, \mathbf{b}^{*}_{<0})  - \Phi(\mathbf{b}^{*}_{=0}, \mathbf{b}^{t+1}_{<0}) \Big)& \leq  \sum_{t = 0}^{T-1} \sum_{i: \rho_i = 0} \KL(\bbb_{i}^* || \bbb_{i}^t) + \sum_{i: -\infty < \rho_i < 0} \frac{\rho_i - 1}{\rho_i}  \KL(\bbb_{i}^* || \bbb_{i}^0)  \\
&\hspace*{1in}+ \sum_{i : \rho_i = -\infty}  \KL(\bbb_{i}^* || \bbb_{i}^0) \\
&\leq   \sum_{i: \rho_i = 0} \KL(\bbb_{i}^* || \bbb_{i}^0) + \sum_{i: -\infty < \rho_i < 0} \frac{\rho_i - 1}{\rho_i}  \KL(\bbb_{i}^* || \bbb_{i}^0)  \\
&\hspace*{1in}+ \sum_{i : \rho_i = -\infty}  \KL(\bbb_{i}^* || \bbb_{i}^0).
\end{align*}
The second inequality holds because of \eqref{eq::zero::compl::cobb}.

By Lemma~\ref{lem::decreasing::compl::cobb},
\begin{align*}
\Phi(\mathbf{b}^{*}_{=0}, \mathbf{b}^{*}_{<0})  - \Phi(\mathbf{b}^{*}_{=0}, \mathbf{b}^{T}_{<0}) &\leq \frac{1}{T}\Big(\sum_{i: \rho_i = 0} \KL(\bbb_{i}^* || \bbb_{i}^0) + \sum_{i: -\infty < \rho_i < 0} \frac{\rho_i - 1}{\rho_i}  \KL(\bbb_{i}^* || \bbb_{i}^0)  \\
&\hspace*{1in}\hspace*{1in}+ \sum_{i : \rho_i = -\infty}  \KL(\bbb_{i}^* || \bbb_{i}^0)\Big).
\end{align*}
\end{proof}
\begin{theorem} \label{thm::linear::cobb::com}
Suppose there is no buyer with a Leontief utility function. Let $\sigma = \min_{i: \rho_i < 0}\left\{\frac{\rho_i - 1}{\rho_i}\right\}$.
Then:
\begin{align*}
\Phi(\mathbf{b}^{*}_{=0}, \mathbf{b}^{*}_{<0})  - \Phi(\mathbf{b}^{*}_{=0}, \mathbf{b}^{T}_{<0}) \leq \frac{\sigma - 1}{\sigma^{T} - 1} \Big(\sum_{i: \rho_i = 0} \KL(\bbb_{i}^* || \bbb_{i}^0) + \sum_{i:  \rho_i < 0} \KL(\bbb_{i}^* || \bbb_{i}^0)\Big).
\end{align*}
\end{theorem}
\begin{proof}
If there is no buyer with a Leontief utility function, then we do not drop $D$ in \eqref{ineq::comp::neg::rhs}. Therefore, instead of \eqref{ineq::comp::neg::upper}, we have:
\begin{align*}
- \Phi(\mathbf{b}^{*}_{=0}, \mathbf{b}^{t+1}_{<0}) &~~~~~~~~~\leq  -\Phi(\mathbf{b}^{*}_{=0}, \mathbf{b}^{*}_{<0}) + \sum_{i: \rho_i = 0} \KL(\bbb_{i}^* || \bbb_{i}^t) \\
&~~~~~~~~~~~~+ \sum_{i: \rho_i < 0}   \KL(\bbb_{i}^* || \bbb_{i}^t)  - \sum_{i:  \rho_i < 0} \frac{\rho_i - 1}{\rho_i}  \KL(\bbb_{i}^* || \bbb_{i}^{t+1}).\numberthis
\end{align*}

Recall $\sigma = \min_{i: \rho_i < 0}\left\{\frac{\rho_i - 1}{\rho_i}\right\}$. Then, multiplying both sides by $\sigma^t$ and summing over $t$ yields:
\begin{align*}
\sum_{t = 0}^{T-1} \sigma^t \Big(\Phi(\mathbf{b}^{*}_{=0}, \mathbf{b}^{*}_{<0})  - \Phi(\mathbf{b}^{*}_{=0}, \mathbf{b}^{t+1}_{<0})\Big) \leq \sum_{i: \rho_i = 0} \KL(\bbb_{i}^* || \bbb_{i}^0) + \sum_{i:  \rho_i < 0} \KL(\bbb_{i}^* || \bbb_{i}^0).
\end{align*}

By Lemma~\ref{lem::decreasing::compl::cobb},
\begin{align*}
\Phi(\mathbf{b}^{*}_{=0}, \mathbf{b}^{*}_{<0})  - \Phi(\mathbf{b}^{*}_{=0}, \mathbf{b}^{T}_{<0}) \leq \frac{\sigma - 1}{\sigma^{T} - 1} \Big(\sum_{i: \rho_i = 0} \KL(\bbb_{i}^* || \bbb_{i}^0) + \sum_{i:  \rho_i < 0} \KL(\bbb_{i}^* || \bbb_{i}^0)\Big).
\end{align*}
\end{proof}

\subsection{Proportional Response with the Entire CES Range}

Formally, in this case, Damped Proportional Response is defined as follows:
\begin{align*}
b_{ij}^{t+1} &= e_i \frac{\Big[b_{ij}^t \cdot a_{ij} \Big(\frac{b^t_{ij}}{p^t_j} \Big)^{\rho_i}\Big]^{\frac{1}{2}}}{\sum_{k}\Big[b_{ik}^t \cdot  a_{ik} \Big( \frac{b^t_{ik}}{p^t_{k}} \Big)^{\rho_i}\Big]^{\frac{1}{2}}}, 
\hhsp\mbox{for $\rho_i > 0$;} 
\hsp\hsp b_{ij}^{t+1} = e_i \frac{\Big[b_{ij}^t \cdot a_{ij}^t\Big]^{\frac{1}{2}}}{\sum_{k}\Big[b_{ik}^t \cdot a_{ik}^t\Big]^{\frac{1}{2}}},\hhsp\mbox{for $\rho_i = 0$;}\\
b_{ij}^{t+1} &= e_i \frac{\Big[b_{ij}^t \cdot \Big(\frac{a_{ij}}{{p^t_{j}}^{\rho_i}}\Big)^{\frac{1}{1 - \rho_i}}\Big]^{\frac{1}{2}}}{\sum_{k}\Big[b_{ik}^t \cdot \Big(\frac{a_{ik}}{{p^t_{k}}^{\rho_i}}\Big)^{\frac{1}{1 - \rho_i}}\Big]^{\frac{1}{2}}},\hhsp\mbox{for $-\infty < \rho_i < 0$;}
\hsp\hsp  b_{ij}^{t+1} = e_i \frac{\Big[b_{ij}^t \cdot \Big(\frac{c_{ij}}{p^t_{j}}\Big)^{-1}\Big]^{\frac{1}{2}}}{\sum_{k}\Big[b_{ik}^t \cdot \Big(\frac{c_{ik}}{p^t_{k}}\Big)^{-1}\Big]^{\frac{1}{2}}},\hhsp\mbox{for $\rho_i = -\infty$;}\\
p_j^{t+1} &=\sum_j b_{ij}^{t+1}&
\end{align*}

Similarly to the Cobb-Douglas-free domain, we have the following observations.
\begin{lemma} \label{lem::HGPRD::eq::repeat}
If $\rho_i > 0$ for buyer $i$, then Damped Proportional Response is equivalent to mirror descent with a halved step size, defined as follows:
\begin{align*}
\bbb_{i}^{t+1} = {\arg \min}_{\bbb_{i} : \sum_j b_{ij} = e_i} \left\{\langle \nabla_{\bbb_{i}}\Phi(\mathbf{b}^t), \bbb_{i} - \bbb_{i}^t\rangle + \frac{2}{\rho_i} \KL(\bbb_{i} || \bbb_{i}^t)\right\};
\end{align*}
if $-\infty < \rho_i < 0$ for buyer $i$, then Damped Proportional Response is equivalent to mirror descent with a halved step size, defined as follows:
\begin{align*}
\bbb_{i}^{t+1} = {\arg \min}_{\bbb_{i} : \sum_j b_{ij} = e_i} \left\{- \langle \nabla_{\bbb_{i}}\Phi(\mathbf{b}^t), \bbb_{i} - \bbb_{i}^t\rangle + \frac{2(\rho_i - 1)}{\rho_i} \KL(\bbb_{i} || \bbb_{i}^t)\right\}.
\end{align*}
and if $\rho_i = -\infty$ for buyer $i$, then Damped Proportional Response is equivalent to mirror descent with a halved step size, defined as follows:
\begin{align*}
\bbb_{i}^{t+1} = {\arg \min}_{\bbb_{i} : \sum_j b_{ij} = e_i} \{- \langle \nabla_{\bbb_{i}}\Phi(\mathbf{b}^t), \bbb_{i} - \bbb_{i}^t\rangle + 2 \KL(\bbb_{i} || \bbb_{i}^t)\}.
\end{align*}
\end{lemma}
\begin{proof}
By calculation.
\end{proof}

However, we find that for the buyers with Cobb-Douglas utility functions, their updating rule cannot be written in mirror descent form. Instead, we make a separate argument for these buyers.

Let $\bbb_{i}^*$ be the equilibrium spending of buyer $i$. If $\rho_i = 0$ for buyer $i$, then her updating rule only depends on her previous spending and her preferences, and it is independent of the other buyers.  
Consequently, as we show in the following lemma, the convergence rate of Damped Proportional Response for the buyers with Cobb-Douglas utilities will be fast.

\begin{lemma} \label{lem::HGPRD::eq::0}
If $\rho_i = 0$ for buyer $i$, then:
\begin{align*}
\KL(\bbb_{i}^* || \bbb_{i}^0) \geq \sum_{t = 1}^{T} \KL(\bbb_{i}^* || \bbb_{i}^t)\hsp\hsp\mbox{and}\hsp\hsp\KL(\bbb_{i}^*|| \bbb_{i}^t) \geq 2 \KL(\bbb_{i}^* || \bbb_{i}^{t+1}).
\end{align*}
\end{lemma}
\begin{proof}
First, we want to show that for any buyer $i$ with $\rho_i = 0$, Damped Proportional Response is equivalent to mirror descent on $\Psi(\bbb_{i}) = - \sum_j b_{ij} \log \frac{a_{ij}}{b_{ij}}$ with halved step size:
\begin{align*}
\bbb_{i}^{t+1} = {\arg \min}_{\bbb_{i} : \sum_j b_{ij} = e_i} \{\langle \nabla \Psi(\bbb_{i}^t), \bbb_{i} - \bbb_{i}^t\rangle + 2 \KL(\bbb_{i} || \bbb_{i}^t)\},
\end{align*}

Note that
\begin{align*}
\nabla_{b_{ij}} \Psi(\bbb_{i}) = - \log \frac{a_{ij}}{b_{ij}} + 1.
\end{align*}
Then, by calculation, $b_{ij}^{t+1} = e_i \frac{( b_{ij}^t \cdot a_{ij})^{\frac{1}{2}}}{\sum_{k}( b_{ik}^t \cdot a_{ik})^{\frac{1}{2}}}$, which is exactly the Damped Proportional Response update rule.

Furthermore, it is easy to see that $\Psi$ is a convex function and it satisfies the following equality:
\begin{align}
\Psi(\bbb_{i}) - \Psi(\bbb_{i}') - \langle \nabla \Psi(b_{i}'), \bbb_{i} - \bbb_{i}' \rangle =  \KL(\bbb_{i}, \bbb_{i}'). 
\label{eqn:Psi-equality}
\end{align}
Therefore, setting $\bbb_{i} = \bbb_{i}^{t+1}$ and $\bbb_{i}' = \bbb_{i}^t$ gives
\begin{align*}
\Psi(\bbb_{i}^{t+1})  - \Psi(\bbb_{i}^t) &= \langle \nabla \Psi(\bbb_{i}^t), \bbb_{i}^{t+1} - \bbb_{i}^t \rangle + 2\KL(\bbb_{i}^{t+1} || \bbb_{i}^t) -  \KL(\bbb_{i}^{t+1} || \bbb_{i}^t) 
\end{align*}
and then by Lemma~\ref{ineq::basic::bregman} with 
$g(\cdot) = \innerprod{\nabla \Psi(\bbb_{i}^t)}{\cdot - \bbb_{i}^t}$,
$x^+ = \bbb_{i}^{t+1}$, $x = \bbb_{i}^* $, $y = \bbb_{i}^t$, and $d(\cdot,\cdot) = 2\KL(\cdot || \cdot)$:
\begin{align}
\label{sadd::sandwich::cobb}
\Psi(\bbb_{i}^{t+1}) - \Psi(\bbb_{i}^t)
&\leq \langle \nabla \Psi(\bbb_{i}^t), \bbb_{i}^* - \bbb_{i}^t \rangle + 2\KL(\bbb_{i}^* || \bbb_{i}^t) - 2\KL(\bbb_{i}^* || \bbb_{i}^{t+1}) - \KL(\bbb_{i}^{t+1} || \bbb_{i}^t).
\end{align}

Setting $\bbb_{i} = \bbb_{i}^*$, $\bbb_{i}' = \bbb_{i}^t$ in \eqref{eqn:Psi-equality} gives:
\begin{align*}
\Psi(\bbb_{i}^t) - \Psi(\bbb_{i}^*) = - \langle \nabla\Psi(\bbb_{i}^t), \bbb_{i}^* - \bbb_{i}^t \rangle - \KL(\bbb_{i}^* || \bbb_{i}^t).
\end{align*}

And combining this with \eqref{sadd::sandwich::cobb} gives:
\begin{align*} \label{ineq::cobb::1}
\Psi(\bbb_{i}^{t+1}) - \Psi(\bbb_{i}^*) \leq \KL(\bbb^*_i || \bbb_{i}^t) - 2 \KL(\bbb_{i}^* || \bbb_{i}^{t+1}). \numberthis
\end{align*}
Since $\Psi(\cdot)$ is a convex function and $\bbb_{i}^*$ is the minimum point for $\Psi$, \eqref{ineq::cobb::1} implies:
\begin{align*} \label{ineq::saddle::zero}
\KL(\bbb^*_i || \bbb_{i}^t) \geq 2 \KL(\bbb_{i}^* || \bbb_{i}^{t+1}). \numberthis
\end{align*}
Note this inequality holds for any $t \geq 0$. So, for any $T$,
\begin{align*}
\KL(\bbb_{i}^* || \bbb_{i}^0) \geq \sum_{t = 1}^{T} \KL(\bbb_{i}^* || \bbb_{i}^t).
\end{align*}
\end{proof}

For the next part, recall that $\mathbf{b}_{>0}$, $\mathbf{b}_{=0}$ and $\mathbf{b}_{<0}$ denote the spending of those buyers with $\rho_i > 0$, $\rho_i = 0$, and $\rho_i < 0$, respectively, and that we rewrote $\Phi(\mathbf{b})$  as $\Phi(\mathbf{b}_{>0}, \mathbf{b}_{=0}, \mathbf{b}_{<0})$.

With a simple calculation one can show:
\begin{itemize}
\item For fixed $\mathbf{b}_{=0}$ and $\mathbf{b}_{<0}$, $\Phi(\cdot, \mathbf{b}_{=0}, \mathbf{b}_{<0})$ is a convex function.
\item For fixed $\mathbf{b}_{>0}$ and $\mathbf{b}_{=0}$, $\Phi(\mathbf{b}_{>0}, \mathbf{b}_{=0}, \cdot)$ is a concave function.
\item Let $\mathbf{b}^*_{>0}$, $\mathbf{b}^*_{=0}$ and $\mathbf{b}^*_{<0}$ be the market equilibrium of the Fisher market; then 
\begin{itemize}
\item $\mathbf{b}^*_{>0}$ minimizes $\Phi(\cdot, \mathbf{b}^*_{=0}, \mathbf{b}^*_{<0})$;
\item $\mathbf{b}^*_{<0}$ maximizes $\Phi(\mathbf{b}^*_{>0}, \mathbf{b}^*_{=0}, \cdot)$.
\end{itemize}
\end{itemize}

\begin{theorem}\label{thm::complete::1T}
Damped Proportional Response converges to the equilibrium with a convergence rate of: 
\begin{align*}
&\sum_{t = 1}^{T} \Bigg[\Phi(\mathbf{b}^{t}_{>0}, \mathbf{b}^{*}_{=0}, \mathbf{b}^{*}_{<0}) - \Phi (\mathbf{b}^{*}_{>0}, \mathbf{b}^{*}_{=0}, \mathbf{b}^{t}_{<0}) \Bigg] \\
 &~~~~~~~~~~~~\leq 4 \sum_{i : \rho_i = 0} \KL(\bbb^*_i || \bbb^0_i)  + \sum_{i: -\infty < \rho_i < 0} \frac{2(\rho_i - 1)}{\rho_i}  \KL(\bbb_{i}^* || \bbb_{i}^0) \\
 &\hspace*{1in}\hspace*{1in}\hspace*{1in}+ \sum_{i: \rho_i > 0} \frac{2}{\rho_i}  \KL(\bbb_{i}^* || \bbb_{i}^0)+ \sum_{i: \rho_i = -\infty} \KL(\bbb_{i}^* || \bbb_{i}^0).
\end{align*}
\end{theorem}
\begin{proof}

First, let's look at $\mathbf{b}^t_{>0}$. By Lemma~\ref{lem::HGPRD::eq::repeat},
we know that for those $i$ for which $\rho_i > 0$,
\begin{align*}
\bbb_{i}^{t+1} = {\arg \min}_{\bbb_{i} : \sum_j b_{ij} = e_i} \left\{\langle \nabla_{\bbb_{i}}\Phi(\mathbf{b}^t), \bbb_{i} - \bbb_{i}^t\rangle + \frac{2}{\rho_i} \KL(\bbb_{i} || \bbb_{i}^t)\right\}.
\end{align*}
 Therefore, by Lemma~\ref{ineq::basic::bregman} with $x^+ = \mathbf{b}^{t+1}_{>0}$, $x = \mathbf{b}^{*}_{>0}$, $y = \mathbf{b}^t_{>0}$,
\begin{align*}
&\langle \nabla_{\mathbf{b}_{>0}} \Phi(\mathbf{b}^t_{>0}, \mathbf{b}^t_{=0}, \mathbf{b}^t_{<0}), \mathbf{b}^{t+1}_{>0} - \mathbf{b}^t_{>0} \rangle + \sum_{i: \rho_i > 0} \frac{2}{\rho_i} \KL(\bbb_{i}^{t+1} || \bbb_{i}^t) \\
&\hspace*{0.2in}\leq \langle \nabla_{\mathbf{b}_{>0}} \Phi(\mathbf{b}^t_{>0}, \mathbf{b}^t_{=0}, \mathbf{b}^t_{<0}), \mathbf{b}^{*}_{>0} - \mathbf{b}^t_{>0} \rangle + \sum_{i: \rho_i > 0} \frac{2}{\rho_i}  \KL(\bbb_{i}^* || \bbb_{i}^t)  - \sum_{i: \rho_i > 0} \frac{2}{\rho_i}  \KL(\bbb_{i}^* || \bbb_{i}^{t+1}).
\end{align*}

This is equivalent to:
\begin{align*} \label{ineq::sadd::pos::triangle}
&\underbrace{\langle \nabla \Phi(\mathbf{b}^t_{>0}, \mathbf{b}^t_{=0}, \mathbf{b}^t_{<0}), (\mathbf{b}^{t+1}_{>0}, \mathbf{b}^{*}_{=0}, \mathbf{b}^{t+1}_{<0}) - (\mathbf{b}^t_{>0},\mathbf{b}^{t}_{=0}, \mathbf{b}^{t}_{<0}) \rangle + \sum_{i: \rho_i > 0} \frac{2}{\rho_i} \KL(\bbb_{i}^{t+1} || \bbb_{i}^t)}_{\texttt{LHS}} \\
&~~~~~~~~~\leq \langle \nabla \Phi(\mathbf{b}^t_{>0}, \mathbf{b}^t_{=0}, \mathbf{b}^t_{<0}), (\mathbf{b}^{*}_{>0}, \mathbf{b}^{*}_{=0}, \mathbf{b}^{t+1}_{<0}) - (\mathbf{b}^t_{>0},\mathbf{b}^{t}_{=0}, \mathbf{b}^{t}_{<0}) \rangle \\
&~~~~~~~~~~~~~~\underbrace{\hspace*{1in}\hspace*{1in}+ \sum_{i: \rho_i > 0} \frac{2}{\rho_i}  \KL(\bbb_{i}^* || \bbb_{i}^t)  - \sum_{i: \rho_i > 0} \frac{2}{\rho_i}  \KL(\bbb_{i}^* || \bbb_{i}^{t+1})}_{\texttt{RHS}}. \numberthis
\end{align*} 

From the second inequality in \eqref{ineq::saddle::upper::lower::2}, the LHS term is lower bounded by:
\begin{align*} \label{ineq::sadd::pos::lhs}
&\Phi (\mathbf{b}^{t+1}_{>0}, \mathbf{b}^{*}_{=0}, \mathbf{b}^{t+1}_{<0}) - \Phi(\mathbf{b}^t_{>0}, \mathbf{b}^t_{=0}, \mathbf{b}^t_{<0}) - \underbrace{\sum_{i: \rho_i \neq \{0, -\infty\}} \frac{1}{\rho_i} \KL(\bbb^{t+1}_i || \bbb^t_i)}_{A} \\
&\hspace*{1in}\hspace*{1in}- \sum_{i: \rho_i = 0} \KL(\bbb^*_i || \bbb^t_i) +  \sum_{i: \rho_i > 0} \frac{2}{\rho_i}\KL(\bbb_{i}^{t+1} || \bbb_{i}^t)  \numberthis
\end{align*}
and from the first inequality, the RHS term is upper bounded by:
\begin{align*} \label{ineq::sadd::pos::rhs}
& \Phi (\mathbf{b}^{*}_{>0}, \mathbf{b}^{*}_{=0}, \mathbf{b}^{t+1}_{<0}) - \Phi(\mathbf{b}^t_{>0}, \mathbf{b}^t_{=0}, \mathbf{b}^t_{<0}) - \underbrace{\sum_{i: \rho_i > 0} \frac{1 - \rho_i}{\rho_i} \KL(\bbb_{i}^* || \bbb_{i}^t)}_{B} \\
&\hspace*{0.2in}~~~~~~~~~~~~ - \sum_{i: -\infty <\rho_i < 0} \frac{1 - \rho_i}{\rho_i} \KL(\bbb_{i}^{t+1} || \bbb_{i}^t) + \sum_{i: \rho_i = -\infty} \KL(\bbb_{i}^{t+1} || \bbb_{i}^t) \\
&\hspace*{1in}\hspace*{1in}+ \sum_{i: \rho_i > 0} \frac{2}{\rho_i}  \KL(\bbb_{i}^* || \bbb_{i}^t)  - \sum_{i: \rho_i > 0} \frac{2}{\rho_i}  \KL(\bbb_{i}^* || \bbb_{i}^{t+1}). \numberthis
\end{align*}

As $\texttt{LHS} \leq \texttt{RHS}$, as the portion of $A$ for $\rho_i < 0$ is negative, and as $B$ is positive, we have
\begin{align*}\label{ineq::saddle::pos::upper}
&\Phi (\mathbf{b}^{t+1}_{>0}, \mathbf{b}^{*}_{=0}, \mathbf{b}^{t+1}_{<0}) + \sum_{i: \rho_i > 0} \frac{1}{\rho_i} \KL(\bbb_{i}^{t+1} || \bbb_{i}^t)  - \sum_{i: \rho_i = 0} \KL(\bbb^*_i || \bbb^t_i)  \\
&\hspace*{0.2in}~~~~~~~~~\leq \Phi (\mathbf{b}^{*}_{>0}, \mathbf{b}^{*}_{=0}, \mathbf{b}^{t+1}_{<0}) + \sum_{i: -\infty< \rho_i < 0} \frac{\rho_i - 1}{\rho_i}\KL(\bbb_{i}^{t+1} || \bbb_{i}^t) + \sum_{i: \rho_i = -\infty} \KL(\bbb_{i}^{t+1} || \bbb_{i}^t) \\
&\hspace*{0.2in}\hspace*{0.2in}~~~~~~~~~~~~~~~+ \sum_{i: \rho_i > 0} \frac{2}{\rho_i}  \KL(\bbb_{i}^* || \bbb_{i}^t)  - \sum_{i: \rho_i > 0} \frac{2}{\rho_i}  \KL(\bbb_{i}^* || \bbb_{i}^{t+1}). \numberthis
\end{align*}

Next, let's look at $\bbb^t_{<0}$. The argument used here is similar to that for $	\bbb^t_{>0}$.  First, by Lemma \ref{lem::HGPRD::eq::repeat}, we know that for those $i$ for which $-\infty < \rho_i < 0$,
\begin{align*}
\bbb_{i}^{t+1} = {\arg \min}_{\bbb_{i} : \sum_j b_{ij} = e_i} \left\{- \langle \nabla_{\bbb_{i}}\Phi(\mathbf{b}^t), \bbb_{i} - \bbb_{i}^t\rangle + \frac{2(\rho_i - 1)}{\rho_i} \KL(\bbb_{i} || \bbb_{i}^t)\right\},
\end{align*}
and for those $i$ for which $\rho_i = -\infty$,
\begin{align*}
\bbb_{i}^{t+1} = {\arg \min}_{\bbb_{i} : \sum_j b_{ij} = e_i} \{- \langle \nabla_{\bbb_{i}}\Phi(\mathbf{b}^t), \bbb_{i} - \bbb_{i}^t\rangle + 2\KL(\bbb_{i} || \bbb_{i}^t)\},
\end{align*}

Therefore, by Lemma~\ref{ineq::basic::bregman}, 
\begin{align*}
&- \langle \nabla_{\mathbf{b}_{<0}} \Phi(\mathbf{b}^t_{>0}, \mathbf{b}^t_{=0}, \mathbf{b}^t_{<0}), \mathbf{b}^{t+1}_{<0} - \mathbf{b}^t_{<0} \rangle +  \sum_{i: -\infty < \rho_i <0}\frac{2(\rho_i - 1)}{\rho_i} \KL(\bbb_{i}^{t+1} || \bbb_{i}^t) + \sum_{i : \rho_i = -\infty} 2 \KL(\bbb_{i}^{t+1} || \bbb_{i}^t) \\
&\hspace*{0.2in}\leq -\langle \nabla_{\mathbf{b}_{<0}} \Phi(\mathbf{b}^t_{>0}, \mathbf{b}^t_{=0}, \mathbf{b}^t_{<0}), \mathbf{b}^{*}_{<0} - \mathbf{b}^t_{<0} \rangle + \sum_{i:  -\infty < \rho_i < 0} \frac{2(\rho_i - 1)}{\rho_i}  \KL(\bbb_{i}^* || \bbb_{i}^t)  \\
&\hspace*{0.2in}\hspace*{0.2in}~~~~~~~~~~~~~~~~~~- \sum_{i: -\infty < \rho_i < 0} \frac{2(\rho_i - 1)}{\rho_i}  \KL(\bbb_{i}^* || \bbb_{i}^{t+1})+ \sum_{i : \rho_i = -\infty} 2\KL(\bbb_{i}^* || \bbb_{i}^t) -\sum_{i : \rho_i = -\infty} 2 \KL(\bbb_{i}^* || \bbb_{i}^{t+1}).
\end{align*}

This is equivalent to:
\begin{align*} \label{ineq::sadd::neg::triangle}
&- \langle \nabla \Phi(\mathbf{b}^t_{>0}, \mathbf{b}^t_{=0}, \mathbf{b}^t_{<0}), (\mathbf{b}^{t+1}_{>0}, \mathbf{b}^{*}_{=0}, \mathbf{b}^{t+1}_{<0}) - (\mathbf{b}^t_{>0},\mathbf{b}^{t}_{=0}, \mathbf{b}^{t}_{<0})  \rangle \\
&\underbrace{+  \sum_{i: -\infty <\rho_i <0}\frac{2(\rho_i - 1)}{\rho_i} \KL(\bbb_{i}^{t+1} || \bbb_{i}^t)+ \sum_{i : \rho_i = -\infty}2 \KL(\bbb_{i}^{t+1} || \bbb_{i}^t)~~~~~~~~~}_{\texttt{LHS}} \\
&\hspace*{0.2in}~~~~~~~~~\leq -\langle \nabla \Phi(\mathbf{b}^t_{>0}, \mathbf{b}^t_{=0}, \mathbf{b}^t_{<0}), (\mathbf{b}^{t+1}_{>0}, \mathbf{b}^{*}_{=0}, \mathbf{b}^{*}_{<0}) - (\mathbf{b}^t_{>0},\mathbf{b}^{t}_{=0}, \mathbf{b}^{t}_{<0})  \rangle \\
&\hspace*{0.2in}\hspace*{0.2in}~~~~~~~~~~~~~~~~~~+ \sum_{i: -\infty < \rho_i < 0} \frac{2(\rho_i - 1)}{\rho_i}  \KL(\bbb_{i}^* || \bbb_{i}^t)  - \sum_{i:  -\infty <\rho_i < 0} \frac{2(\rho_i - 1)}{\rho_i}  \KL(\bbb_{i}^* || \bbb_{i}^{t+1})\\
&\hspace*{0.2in}~~~~~~~~~~~~\underbrace{\hspace*{0.2in}~~~~~~+ \sum_{i : \rho_i = -\infty} 2 \KL(\bbb_{i}^* || \bbb_{i}^t) -\sum_{i : \rho_i = -\infty} 2 \KL(\bbb_{i}^* || \bbb_{i}^{t+1}).\hspace*{0.2in}\hspace*{0.2in}\hspace*{0.2in}~~~~~~~~~~~~~~~~~~~~~~~~~~~~}_{\texttt{RHS}} \numberthis
\end{align*}

From the first inequality in \eqref{ineq::saddle::upper::lower::2}, the LHS term is lower bounded by:
\begin{align*} \label{ineq::sadd::neg::lhs}
& - \Phi(\mathbf{b}^{t+1}_{>0}, \mathbf{b}^{*}_{=0}, \mathbf{b}^{t+1}_{<0}) + \Phi(\mathbf{b}^{t}_{>0}, \mathbf{b}^{t}_{=0}, \mathbf{b}^{t}_{<0}) - \underbrace{\sum_{i: \rho_i \neq \{0, -\infty\}} \frac{\rho_i - 1}{\rho_i} \KL(\bbb_{i}^{t+1} || \bbb_{i}^t)}_{C}  - \sum_{i: \rho_i = -\infty}\KL(\bbb_{i}^{t+1} || \bbb_{i}^t) \\
& \hspace*{0.2in}~~~~~~~~~+  \sum_{i: \rho_i <0}\frac{2(\rho_i - 1)}{\rho_i} \KL(\bbb_{i}^{t+1} || \bbb_{i}^t) +  \sum_{i: \rho_i = -\infty} 2 \KL(\bbb_{i}^{t+1} || \bbb_{i}^t) \numberthis
\end{align*}
and from the second inequality, the RHS term is upper bounded by:
\begin{align*} \label{ineq::sadd::neg::rhs}
 & -\Phi(\mathbf{b}^{t+1}_{>0}, \mathbf{b}^{*}_{=0}, \mathbf{b}^{*}_{<0}) + \Phi(\mathbf{b}^{t}_{>0}, \mathbf{b}^{t}_{=0}, \mathbf{b}^{t}_{<0})  + \sum_{i: \rho_i > 0} \frac{1}{\rho_i}  \KL(\bbb_{i}^{t+1} || \bbb_{i}^t) + \underbrace{\sum_{i: -\infty < \rho_i < 0} \frac{1}{\rho_i} \KL(\bbb_{i}^* || \bbb_{i}^t)}_{D}  \\
&\hspace*{0.2in}~~~~~~~~~+ \sum_{i: \rho_i = 0} \KL(\bbb_{i}^* || \bbb_{i}^t) + \sum_{i: -\infty < \rho_i < 0} \frac{2(\rho_i - 1)}{\rho_i}  \KL(\bbb_{i}^* || \bbb_{i}^t)  - \sum_{i: -\infty < \rho_i < 0} \frac{2(\rho_i - 1)}{\rho_i}  \KL(\bbb_{i}^* || \bbb_{i}^{t+1})\\
&\hspace*{0.2in}~~~~~~~~~+ \sum_{i : \rho_i = -\infty} 2 \KL(\bbb_{i}^* || \bbb_{i}^t) -\sum_{i : \rho_i = -\infty} 2 \KL(\bbb_{i}^* || \bbb_{i}^{t+1}). \numberthis
\end{align*}

As $\texttt{LHS} \leq \texttt{RHS}$, as the portion of $C$ for $\rho_i >0$ is negative, and as $D$ is negative, we have:
\begin{align*}\label{ineq::saddle::neg::upper}
&- \Phi(\mathbf{b}^{t+1}_{>0}, \mathbf{b}^{*}_{=0}, \mathbf{b}^{t+1}_{<0}) + \sum_{i: -\infty < \rho_i < 0} \frac{\rho_i - 1}{\rho_i} \KL(\bbb_{i}^{t+1} || \bbb_{i}^t) + \sum_{i: \rho_i = -\infty} \KL(\bbb_{i}^{t+1} || \bbb_{i}^t) \\
&~~~~~~~~~\leq  -\Phi(\mathbf{b}^{t+1}_{>0}, \mathbf{b}^{*}_{=0}, \mathbf{b}^{*}_{<0}) + \sum_{i: \rho_i > 0} \frac{1}{\rho_i}  \KL(\bbb_{i}^{t+1} || \bbb_{i}^t) + \sum_{i: \rho_i = 0} \KL(\bbb_{i}^* || \bbb_{i}^t) \\
&~~~~~~~~~~~~+ \sum_{i: -\infty < \rho_i < 0} \frac{2(\rho_i - 1)}{\rho_i}  \KL(\bbb_{i}^* || \bbb_{i}^t)  - \sum_{i: -\infty < \rho_i < 0} \frac{2(\rho_i - 1)}{\rho_i}  \KL(\bbb_{i}^* || \bbb_{i}^{t+1}) \\
&~~~~~~~~~~~~+ \sum_{i : \rho_i = -\infty} 2 \KL(\bbb_{i}^* || \bbb_{i}^t) -\sum_{i : \rho_i = -\infty} 2 \KL(\bbb_{i}^* || \bbb_{i}^{t+1}).\numberthis
\end{align*}

Summing \eqref{ineq::saddle::pos::upper} and \eqref{ineq::saddle::neg::upper} gives:
\begin{align*}
&\Phi(\mathbf{b}^{t+1}_{>0}, \mathbf{b}^{*}_{=0}, \mathbf{b}^{*}_{<0}) - \Phi (\mathbf{b}^{*}_{>0}, \mathbf{b}^{*}_{=0}, \mathbf{b}^{t+1}_{<0}) \\
&\hspace*{1in}\leq 2  \sum_{i: \rho_i = 0} \KL(\bbb_{i}^* || \bbb_{i}^t) + \sum_{i: -\infty < \rho_i < 0} \frac{2(\rho_i - 1)}{\rho_i}  \KL(\bbb_{i}^* || \bbb_{i}^t)  - \sum_{i: -\infty < \rho_i < 0} \frac{2(\rho_i - 1)}{\rho_i}  \KL(\bbb_{i}^* || \bbb_{i}^{t+1}) \\
&\hspace*{1in}\hspace*{0.2in}~~~~~~~~~+ \sum_{i: \rho_i > 0} \frac{2}{\rho_i}  \KL(\bbb_{i}^* || \bbb_{i}^t)  - \sum_{i: \rho_i > 0} \frac{2}{\rho_i}  \KL(\bbb_{i}^* || \bbb_{i}^{t+1})\\
&\hspace*{1in}\hspace*{0.2in}~~~~~~~~~+ \sum_{i : \rho_i = -\infty} 2 \KL(\bbb_{i}^* || \bbb_{i}^t) -\sum_{i : \rho_i = -\infty} 2 \KL(\bbb_{i}^* || \bbb_{i}^{t+1}). 
\end{align*}

Summing over all $t$ yields:
\begin{align*}
&\sum_{t = 0}^{T-1} \Phi(\mathbf{b}^{t+1}_{>0}, \mathbf{b}^{*}_{=0}, \mathbf{b}^{*}_{<0}) - \Phi (\mathbf{b}^{*}_{>0}, \mathbf{b}^{*}_{=0}, \mathbf{b}^{t+1}_{<0})  \\
&~~~~~~~~~~~~\leq 2 \sum_{t = 0}^{T-1} \sum_{i : \rho_i = 0} \KL(\bbb^*_i || \bbb^t_i)  + \sum_{i: -\infty < \rho_i < 0} \frac{2(\rho_i - 1)}{\rho_i}  \KL(\bbb_{i}^* || \bbb_{i}^0) \\
&\hspace*{1in}\hspace*{1in}+ \sum_{i: \rho_i > 0} \frac{2}{\rho_i}  \KL(\bbb_{i}^* || \bbb_{i}^0) + \sum_{i: \rho_i = -\infty} \KL(\bbb_{i}^* || \bbb_{i}^0) \\
&~~~~~~~~~~~~\leq 4 \sum_{i : \rho_i = 0} \KL(\bbb^*_i || \bbb^0_i)  + \sum_{i: -\infty < \rho_i < 0} \frac{2(\rho_i - 1)}{\rho_i}  \KL(\bbb_{i}^* || \bbb_{i}^0) \\
&\hspace*{1in}\hspace*{1in}+ \sum_{i: \rho_i > 0} \frac{2}{\rho_i}  \KL(\bbb_{i}^* || \bbb_{i}^0)+ \sum_{i: \rho_i = -\infty} \KL(\bbb_{i}^* || \bbb_{i}^0),
\end{align*}
where Lemma~\ref{lem::HGPRD::eq::0} is used in bounding the first sum on the right hand side.
\end{proof}

\begin{theorem}\label{thm::complete::linear}
Suppose there is no buyer with either a linear utility or a Leontief utility. Let 
\begin{align*}
\sigma = \min \left\{\min_{i:\rho_i > 0} \left\{\frac{2}{1+\rho_i}\right\}\hhsp,\hhsp\min_{i:\rho_i < 0} \left\{\frac{2(\rho_i - 1)}{2\rho_i - 1}\right\}\right\} \hsp\hsp\mbox{(and so $1 < \sigma < 2$)};
\end{align*}
 then Damped Proportional Response converges to the equilibrium with a convergence rate of:
\begin{align*}
 \Phi(\mathbf{b}^{T}_{>0}, \mathbf{b}^{*}_{=0}, \mathbf{b}^{*}_{<0}) - \Phi (\mathbf{b}^{*}_{>0}, \mathbf{b}^{*}_{=0}, \mathbf{b}^{T}_{<0}) &\leq \frac{1}{\sigma^{T-1}} \Bigg[ \frac{4}{2 - \sigma}   \sum_{i: \rho_i = 0} \KL(\bbb_{i}^* || \bbb_{i}^0) \\
&~~~~~~~~~+ \sum_{i: \rho_i < 0} \frac{2\rho_i - 1}{\rho_i}  \KL(\bbb_{i}^* || \bbb_{i}^0)  + \sum_{i: \rho_i > 0} \frac{1+\rho_i}{\rho_i}  \KL(\bbb_{i}^* || \bbb_{i}^0)\Bigg]. 
\end{align*} 
\end{theorem}
\begin{proof}
In this case, on combining \eqref{ineq::sadd::pos::lhs} and \eqref{ineq::sadd::pos::rhs}, \eqref{ineq::saddle::pos::upper} is changed to:
\begin{align*}\label{ineq::saddle::pos::strong}
&\Phi (\mathbf{b}^{t+1}_{>0}, \mathbf{b}^{*}_{=0}, \mathbf{b}^{t+1}_{<0}) + \sum_{i: \rho_i > 0} \frac{1}{\rho_i} \KL(\bbb_{i}^{t+1} || \bbb_{i}^t)  - \sum_{i: \rho_i = 0} \KL(b^*_i || b^t_i)  \\
&~~~~~~~~~\leq \Phi (\mathbf{b}^{*}_{>0}, \mathbf{b}^{*}_{=0}, \mathbf{b}^{t+1}_{<0}) + \sum_{i: \rho_i < 0} \frac{\rho_i - 1}{\rho_i}\KL(\bbb_{i}^{t+1} || \bbb_{i}^t) + \sum_{i: \rho_i > 0} \frac{1 + \rho_i}{\rho_i}  \KL(\bbb_{i}^* || \bbb_{i}^t)  - \sum_{i: \rho_i > 0} \frac{2}{\rho_i}  \KL(\bbb_{i}^* || \bbb_{i}^{t+1}). \numberthis
\end{align*}

Also, on combining \eqref{ineq::sadd::neg::lhs} and \eqref{ineq::sadd::neg::rhs}, \eqref{ineq::saddle::neg::upper} is changed to:
\begin{align*}\label{ineq::saddle::neg::strong}
&- \Phi(\mathbf{b}^{t+1}_{>0}, \mathbf{b}^{*}_{=0}, \mathbf{b}^{t+1}_{<0}) + \sum_{i: \rho_i < 0} \frac{\rho_i - 1}{\rho_i} \KL(\bbb_{i}^{t+1} || \bbb_{i}^t) \\
&\hspace*{0.2in}\leq  -\Phi(\mathbf{b}^{t+1}_{>0}, \mathbf{b}^{*}_{=0}, \mathbf{b}^{*}_{<0}) + \sum_{i: \rho_i > 0} \frac{1}{\rho_i}  \KL(\bbb_{i}^{t+1} || \bbb_{i}^t) + \sum_{i: \rho_i = 0} \KL(\bbb_{i}^* || \bbb_{i}^t) \\
&\hspace*{0.2in}+ \sum_{i: \rho_i < 0} \frac{2\rho_i - 1}{\rho_i}  \KL(\bbb_{i}^* || \bbb_{i}^t)  - \sum_{i: \rho_i < 0} \frac{2(\rho_i - 1)}{\rho_i}  \KL(\bbb_{i}^* || \bbb_{i}^{t+1}).\numberthis
\end{align*}

Combining \eqref{ineq::saddle::pos::strong} and \eqref{ineq::saddle::neg::strong} yields:
\begin{align*}
\Phi(\mathbf{b}^{t+1}_{>0}, \mathbf{b}^{*}_{=0}, \mathbf{b}^{*}_{<0}) - \Phi (\mathbf{b}^{*}_{>0}, \mathbf{b}^{*}_{=0}, \mathbf{b}^{t+1}_{<0}) &\leq 2  \sum_{i: \rho_i = 0} \KL(\bbb_{i}^* || \bbb_{i}^t) \\
&~~~~~~~~~+ \sum_{i: \rho_i < 0} \frac{2\rho_i - 1}{\rho_i}  \KL(\bbb_{i}^* || \bbb_{i}^t)  - \sum_{i: \rho_i < 0} \frac{2(\rho_i - 1)}{\rho_i}  \KL(\bbb_{i}^* || \bbb_{i}^{t+1}) \\
&~~~~~~~~~+ \sum_{i: \rho_i > 0} \frac{1+\rho_i}{\rho_i}  \KL(\bbb_{i}^* || \bbb_{i}^t)  - \sum_{i: \rho_i > 0} \frac{2}{\rho_i}  \KL(\bbb_{i}^* || \bbb_{i}^{t+1}). 
\end{align*}

Recall that $\sigma = \min \left\{\min_{i:\rho_i > 0} \left\{\frac{2}{1+\rho_i}\right\}\hhsp,\hhsp\min_{i:\rho_i < 0} \left\{\frac{2(\rho_i - 1)}{2\rho_i - 1}\right\}\right\}$; then,
\begin{align*}
&\sum_{t = 0}^{T-1} \sigma^t \Big( \Phi(\mathbf{b}^{t+1}_{>0}, \mathbf{b}^{*}_{=0}, \mathbf{b}^{*}_{<0}) - \Phi (\mathbf{b}^{*}_{>0}, \mathbf{b}^{*}_{=0}, \mathbf{b}^{t+1}_{<0})\Big) \\
&\hspace*{1in}\leq 2  \sum_{t = 0}^{T-1} \sum_{i: \rho_i = 0} \sigma^t \cdot \KL(\bbb_{i}^* || \bbb_{i}^t) \\
&\hspace*{1in}~~~~~~~~~+ \sum_{i: \rho_i < 0} \frac{2\rho_i - 1}{\rho_i}  \KL(\bbb_{i}^* || \bbb_{i}^0)  + \sum_{i: \rho_i > 0} \frac{1+\rho_i}{\rho_i}  \KL(\bbb_{i}^* || \bbb_{i}^0). 
\end{align*}

By \eqref{ineq::saddle::zero} and $1< \sigma < 2$,
\begin{align*}
&\sum_{t = 0}^{T-1} \sigma^t \Big( \Phi(\mathbf{b}^{t+1}_{>0}, \mathbf{b}^{*}_{=0}, \mathbf{b}^{*}_{<0}) - \Phi (\mathbf{b}^{*}_{>0}, \mathbf{b}^{*}_{=0}, \mathbf{b}^{t+1}_{<0})\Big) \\
&\hspace*{1in}\leq \frac{4}{2 - \sigma}   \sum_{i: \rho_i = 0} \KL(\bbb_{i}^* || \bbb_{i}^0) \\
&\hspace*{1in}~~~~~~~~~+ \sum_{i: \rho_i < 0} \frac{2\rho_i - 1}{\rho_i}  \KL(\bbb_{i}^* || \bbb_{i}^0)  + \sum_{i: \rho_i > 0} \frac{1+\rho_i}{\rho_i}  \KL(\bbb_{i}^* || \bbb_{i}^0). 
\end{align*}

Therefore,
\begin{align*}
 \Phi(\mathbf{b}^{T}_{>0}, \mathbf{b}^{*}_{=0}, \mathbf{b}^{*}_{<0}) - \Phi (\mathbf{b}^{*}_{>0}, \mathbf{b}^{*}_{=0}, \mathbf{b}^{T}_{<0}) &\leq \frac{1}{\sigma^{T-1}} \Bigg[ \frac{4}{2 - \sigma}   \sum_{i: \rho_i = 0} \KL(\bbb_{i}^* || \bbb_{i}^0) \\
&~~~~~~~~~+ \sum_{i: \rho_i < 0} \frac{2\rho_i - 1}{\rho_i}  \KL(\bbb_{i}^* || \bbb_{i}^0)  + \sum_{i: \rho_i > 0} \frac{1+\rho_i}{\rho_i}  \KL(\bbb_{i}^* || \bbb_{i}^0)\Bigg]. 
\end{align*}
\end{proof}
\section{Relationship between the Eisenberg-Gale Program and our Potential Function}\label{sec:Egand-ourPotFn}

In this section, we show a relationship between our potential function and the objective functions in the Eisenberg-Gale convex program and its dual program.

Let $u_i(\bbx_i)$ be the utility of buyer $i$ when the allocation is $\bbx_i$. Note that $u_i(\bbx_i) = (\sum_j a_{ij}\cdot x_{ij}^{\rho_i})^{\frac{1}{\rho_i}}$ for $1 \geq \rho_i > 0$ and $0 > \rho_i > -\infty$.  For $\rho_i = 0$, $u_i(\bbx_i) = \prod_j x_{ij}^{a_{ij}}$ with $\sum_j a_{ij} = 1$. For $\rho_i = -\infty$, $u_i(\bbx_i) = \min_j \Big\{\frac{x_{ij}}{c_{ij}}\Big\}$.
Our potential function is:
\begin{align*}
&p_j(\mathbf{b}) = \textstyle{\sum_{i}} b_{ij}, \\
&\Phi(\mathbf{b}) = - \sum_{i: \rho_i \neq \{0, -\infty\}} \frac{1}{\rho_i} \sum_{j} b_{ij} \log \frac{a_{ij} b_{ij}^{\rho_i - 1}}{[p_j(\mathbf{b})]^{\rho_i}} - \sum_{i: \rho_i = -\infty} \sum_j b_{ij} \log \frac{b_{ij}}{c_{ij}p_j(\mathbf{b})} + \sum_{i: \rho_i = 0}\sum_j b_{ij} \log p_j(\mathbf{b}).
\end{align*}
Recall that the goal is to minimize $\Phi(\mathbf{b})$ in the substitutes domain and maximize $\Phi(\mathbf{b})$ in the complementary domain.

The objective function for the Eisenberg-Gale program is:
\begin{align*}
\Psi(\mathbf{x}) = \sum_{i} e_i \log u_i(\bbx_i).  
\end{align*}
Recall that the goal is to maximize $\Psi(\mathbf{x})$.

The  objective function for the dual of the Eisenberg-Gale convex program is:
\begin{align*}
\Upsilon(\mathbf{p}) = \max_{\mathbf{x}} \left(\sum_i e_i \log u_i(\bbx_i) + \sum_j p_j \left(1 - \sum_i x_{ij}\right)\right).
\end{align*}
Recall that the goal is to minimize $\Upsilon(\mathbf{p})$.

\subsection{Substitutes Domain}
In the substitutes domain ($\rho_i \geq 0$), let $\mathbf{b}$ be the spending of the buyers; recall that it satisfies $\sum_j b_{ij} = e_i$ for all $i$.  We consider the corresponding allocation $\mathbf{x}(\mathbf{b})$ and the Eisenberg-Gale program, in which $x_{ij} = \frac{b_{ij}}{p_j(\mathbf{b})}$ and $p_j(\mathbf{b}) = \sum_h b_{hj}$. We have the following result.

\begin{theorem}
\begin{align*}
\Psi(\mathbf{x}(\mathbf{b}^*_{>0}, \mathbf{b}^*_{=0})) - \Psi(\mathbf{x}(\mathbf{b}_{>0}, \mathbf{b}^*_{=0})) \leq \Phi(\mathbf{b}_{>0}, \mathbf{b}^*_{=0}) -\Phi(\mathbf{b}^*_{>0}, \mathbf{b}^*_{=0}).
\end{align*}
\end{theorem}
\begin{proof}

First, as $\sum_j \bij= e_i$, using the concavity of the $\log$ function yields:
\begin{align*}
\Psi(\mathbf{x}(\mathbf{b})) &= \sum_{i : \rho_i > 0} \frac{e_i}{\rho_i} \log \Bigg(\sum_j a_{ij} \Big(\frac{b_{ij}}{\mathbf{p}(\mathbf{b})}\Big)^{\rho_i}\Bigg)  + \sum_{i: \rho_i = 0} \sum_j e_i a_{ij} \log \frac{b_{ij}}{\mathbf{p}(\mathbf{b})} \\
&\geq \sum_{i: \rho_i > 0} \frac{e_i}{\rho_i} \sum_j \frac{b_{ij}}{e_i} \log \frac{a_{ij}b_{ij}^{\rho_i - 1} e_i}{(\mathbf{p}(\mathbf{b}))^{\rho_i}} + \sum_{i: \rho_i = 0} \sum_j e_i a_{ij} \log \frac{b_{ij}}{\mathbf{p}(\mathbf{b})}\\  
&= \sum_{ij: \rho_i > 0} \frac{b_{ij}}{\rho_i} \log \frac{a_{ij}b_{ij}^{\rho_i - 1}}{(\mathbf{p}(\mathbf{b}))^{\rho_i}} + \sum_{i: \rho_i > 0} \frac{e_i}{\rho_i} \log e_i + \sum_{i: \rho_i = 0} \sum_j e_i a_{ij} \log \frac{b_{ij}}{\mathbf{p}(\mathbf{b})}.
\end{align*}
Note that if  for each $i$ such that $\rho_i > 0$, $a_{ij}\frac{b_{ij}^{\rho_i - 1}}{\sum_h b_{hj}^{\rho_i}}$ 
are the same for all $j$ with $b_{ij} > 0$, then the inequality above will become an equality. 
Also, at the market equilibrium $\mathbf{b}^*$, this condition holds.
Therefore,
\begin{align*}
&\Psi(\mathbf{x}(\mathbf{b}^*_{>0}, \mathbf{b}_{=0}^*)) - \Psi(\mathbf{x}(\mathbf{b}_{>0}, \mathbf{b}_{=0}^*)) \\
&\hspace*{0.2in}\leq \sum_{ij: \rho_i > 0} \frac{b^*_{ij}}{\rho_i} \log \frac{a_{ij}{b^*_{ij}}^{\rho_i - 1}}{(\mathbf{p}(\mathbf{b}^*_{>0}, \mathbf{b}_{=0}^*))^{\rho_i}} + \sum_{i: \rho_i > 0} \frac{e_i}{\rho_i} \log e_i + \sum_{i: \rho_i = 0} \sum_j e_i a_{ij} \log \frac{b^*_{ij}}{\mathbf{p}(\mathbf{b}^*_{>0}, \mathbf{b}_{=0}^*)} \\
&\hspace*{0.2in}\hspace*{0.2in}- \sum_{ij: \rho_i > 0} \frac{b_{ij}}{\rho_i} \log \frac{a_{ij}b_{ij}^{\rho_i - 1}}{(\mathbf{p}(\mathbf{b}_{>0}, \mathbf{b}_{=0}^*))^{\rho_i}} - \sum_{i: \rho_i > 0} \frac{e_i}{\rho_i} \log e_i - \sum_{i: \rho_i = 0} \sum_j e_i a_{ij} \log \frac{b^*_{ij}}{\mathbf{p}(\mathbf{b}_{>0}, \mathbf{b}_{=0}^*)}\\
&\hspace*{0.2in}=-\Phi(\mathbf{b}^*_{>0}, \mathbf{b}_{=0}^*) +\Phi(\mathbf{b}_{>0}, \mathbf{b}_{=0}^*).
\end{align*}
The last equality follows because $b^*_{ij} = e_i a_{ij}$ for $\rho_i = 0$.
\end{proof}

\subsection{Complementary Domain}
In the complementary domain ($\rho_i \leq 0$),  again $\mathbf{b}$ satisfies $\sum_j b_{ij} = e_i$ for all $i$.  Here we consider the corresponding price $\mathbf{p}(\mathbf{b})$ and the dual of the Eisenberg-Gale program, in which $p_j(\mathbf{b})= \sum_h b_{hj}$. We have the following result.
\begin{theorem}
\begin{align*}
\Upsilon(\mathbf{p}(\mathbf{b}_{=0}^*, \mathbf{b}_{<0})) - \Upsilon(\mathbf{p}(\mathbf{b}_{=0}^*, \mathbf{b}^*_{<0})) \leq \Phi(\mathbf{b}_{=0}^*, \mathbf{b}^*_{<0}) - \Phi(\mathbf{b}_{=0}^*, \mathbf{b}_{<0}).
\end{align*}
\end{theorem}
\begin{proof}
In the proof of Lemma~$5.1$ in \cite{Cole:2016:LMG:2940716.2940720},
 it was shown that the maximum point 
$\mathbf{x}$ in $\Upsilon(\mathbf{p})$ satisfies $\sum_j x_{ij} p_j = e_i$ for all $i$. Therefore,
\begin{align*}
\Upsilon(\mathbf{p}) = \max_{\mathbf{x}: \forall i(\sum_j x_{ij} p_j = e_i)} &\sum_{i: 0> \rho_i > -\infty} e_i \log \Big( \sum_j a_{ij} x_{ij}^{\rho_i}\Big)^{\frac{1}{\rho_i}} + \sum_{i: \rho_i = -\infty} e_i \log \min_j \Big\{\frac{x_{ij}}{c_{ij}}\Big\} \\
&\hspace*{1in}+ \sum_{i: \rho_i = 0} \sum_j e_i a_{ij} \log x_{ij}  + \sum_j p_j - \sum_i e_i.
\end{align*}
Let $b_{ij} = x_{ij} p_j$. Then, 
\begin{align*}
\Upsilon(\mathbf{p}) = \max_{\mathbf{b}: \forall i(\sum_j b_{ij} = e_i)} &\sum_{i: 0 > \rho_i > -\infty} e_i \log \Big( \sum_j a_{ij} \Big(\frac{b_{ij}}{p_j}\Big)^{\rho_i}\Big)^{\frac{1}{\rho_i}} + \sum_{i: \rho_i  = -\infty} e_i \log \min \Big\{\frac{b_{ij}}{p_j c_{ij}} \Big\} \\
&\hspace*{1in}+ \sum_{i: \rho_i = 0} \sum_j e_i a_{ij} \log \frac{b_{ij}}{p_{j}}+ \sum_j p_j - \sum_i e_i.
\end{align*}
Let $\mathbf{b}(\mathbf{p})$ be the spending that maximizes
 $\sum_{i: 0 > \rho_i > -\infty} e_i \log \Big( \sum_j a_{ij} \Big(\frac{b_{ij}}{p_j}\Big)^{\rho_i}\Big)^{\frac{1}{\rho_i}} + \sum_{i: \rho_i  = -\infty} e_i \log \min \Big\{\frac{b_{ij}}{p_j c_{ij}} \Big\}+ \sum_{i: \rho_i = 0} \sum_j e_i a_{ij} \log \frac{b_{ij}}{p_{j}}$ under the constraint $\forall i (\sum_j b_{ij} = e_i)$, so 
\begin{align*} \label{eq::dual::3}
\Upsilon(\mathbf{p}) = &\sum_{i: 0 > \rho_i > -\infty} e_i \log \Big( \sum_j a_{ij} \Big(\frac{b_{ij}(\mathbf{p})}{p_j}\Big)^{\rho_i}\Big)^{\frac{1}{\rho_i}} + \sum_{i: \rho_i  = -\infty} e_i \log \min \Big\{\frac{b_{ij}(\mathbf{p})}{p_j c_{ij}} \Big\} + \sum_{i: \rho_i = 0} \sum_j e_i a_{ij} \log \frac{b_{ij}(\mathbf{p})}{p_{j}} \\
&\hspace*{0.6in}+ \sum_j p_j - \sum_i e_i\\
= &\sum_{i: 0 > \rho_i > -\infty} e_i \log \Big( \sum_j a_{ij} \Big(\frac{b_{ij}(\mathbf{p})}{p_j}\Big)^{\rho_i}\Big)^{\frac{1}{\rho_i}} + \sum_{i: \rho_i  = -\infty} e_i \log \min \Big\{\frac{b_{ij}(\mathbf{p})}{p_j c_{ij}} \Big\} + \sum_{i: \rho_i = 0} \sum_j b_{ij}^* \log \frac{b^*_{ij}}{p_{j}} \\
&\hspace*{0.6in}+ \sum_j p_j - \sum_i e_i. \numberthis
\end{align*}
The second equality holds because, for those buyers with Cobb-Douglas utility functions, their optimal spending is always equal to $b_{ij}^* = e_i a_{ij}$ which is independent of the prices.
With a simple calculation, one can show the following:
\begin{enumerate}
\item For those $i$ such that $ 0 > \rho_i > -\infty$, $a_{ij} \frac{b_{ij}(\mathbf{p})^{\rho_i - 1}}{p_j^{\rho_i}}$ 
are the same for different $j$ by the definition of $\mathbf{b}(\mathbf{p})$. Therefore, 
\begin{align*}\label{eq::dual::1}
e_i \log \Big( \sum_j a_{ij} \Big(\frac{b_{ij}(\mathbf{p})}{p_j}\Big)^{\rho_i}\Big)^{\frac{1}{\rho_i}} &= \frac{e_i}{\rho_i} \sum_j \frac{b_{ij}(\mathbf{p})}{e_i} \log  e_i a_{ij} \frac{b_{ij}(\mathbf{p})^{\rho_i - 1}}{p_j^{\rho_i}} \\
&= \frac{1}{\rho_i} \sum_j b_{ij}(\mathbf{p}) \log  a_{ij} \frac{b_{ij}(\mathbf{p})^{\rho_i - 1}}{p_j^{\rho_i}} + \frac{e_i}{\rho_i} \log e_i. \numberthis
\end{align*}

For those $i$ such that $\rho_i = -\infty$, $\frac{b_{ij}(p)}{p_j c_{ij}}$ 
are the same for different $j$ again by the definition of $\mathbf{b}(\mathbf{p})$. Therefore,
\begin{align*} \label{eq::dual::2}
 e_i \log \min \Big\{\frac{b_{ij}(\mathbf{p})}{p_j c_{ij}} \Big\}  = \sum_j b_{ij}(\mathbf{p}) \log  \frac{b_{ij}(\mathbf{p})}{p_j c_{ij}}. \numberthis
\end{align*}
\item For those $i$ such that $0 > \rho_i > -\infty$, we focus on the function $\frac{1}{\rho_i}  b_{ij} \log  a_{ij} \frac{b_{ij}^{\rho_i - 1}}{p_j^{\rho_i}} $. 
By calculation, given $\mathbf{p}$, this function is a convex function. 
In addition, the minimal point $\mathbf{b}_i$ of the function under the constraint $\sum_j b_{ij} = e_i$ is $\mathbf{b}_i(\mathbf{p})$.  Therefore, combining with \eqref{eq::dual::1} yields
\begin{align*}\label{eq::dual::4}
e_i \log \Big( \sum_j a_{ij} \Big(\frac{b_{ij}(\mathbf{p})}{p_j}\Big)^{\rho_i}\Big)^{\frac{1}{\rho_i}} &\leq \frac{1}{\rho_i} \sum_j b_{ij} \log  a_{ij} \frac{b_{ij}^{\rho_i - 1}}{p_j^{\rho_i}} + \frac{e_i}{\rho_i} \log e_i. \numberthis
\end{align*}
The inequality becomes an equality if  $\mathbf{b}_i = \mathbf{b}_i(\mathbf{p})$.

For those $i$ such that $\rho_i = -\infty$, we focus on the function $\sum_j b_{ij}\log  \frac{b_{ij}}{p_j c_{ij}}$. 
Again, given $\mathbf{p}$, this function is a convex function, and the minimal point $\mathbf{b}_i$ of the function under the constraint $\sum_j b_{ij} = e_i$ is $\mathbf{b}_i(\mathbf{p})$. Therefore,
\begin{align*}\label{eq::dual::5}
e_i \log \min \Big\{\frac{b_{ij}(\mathbf{p})}{p_j c_{ij}} \Big\} \leq \sum_j b_{ij}\log  \frac{b_{ij}}{p_j c_{ij}}. \numberthis
\end{align*}
Also, the inequality becomes an equality if $\mathbf{b}_i = \mathbf{b}_i(\mathbf{p})$.
\end{enumerate}
 Combining \eqref{eq::dual::3}, \eqref{eq::dual::4} and \eqref{eq::dual::5}
\begin{align*}
\Upsilon(\mathbf{p}(\mathbf{b}^*_{=0}, \mathbf{b}_{<0}))&\leq \sum_{i : 0 > \rho_i > -\infty}  \Bigg( \frac{1}{\rho_i} \sum_j b_{ij} \log  a_{ij} \frac{b_{ij}^{\rho_i - 1}}{(p_j(\mathbf{b}^*_{=0}, \mathbf{b}_{<0}))^{\rho_i}}+ \frac{e_i}{\rho_i} \log e_i \Bigg) \\
&\hspace*{0.6in}+ \sum_{i: \rho_i = -\infty} \sum_j b_{ij}\log  \frac{b_{ij}}{p_j(\mathbf{b}^*_{=0}, \mathbf{b}_{<0}) c_{ij}} \\
&\hspace*{0.6in} + \sum_{i: \rho_i = 0} \sum_j b_{ij}^* \log \frac{b^*_{ij}}{p_j(\mathbf{b}^*_{=0}, \mathbf{b}_{<0})} + \sum_{j} p_{j}(\mathbf{b}^*_{=0}, \mathbf{b}_{<0}) - \sum_i e_i\\
&= -\Phi(\mathbf{b}^*_{=0}, \mathbf{b}_{<0}) +  \sum_{j} p_{j}(\mathbf{b}^*_{=0}, \mathbf{b}_{<0})  + \sum_{i : \rho_i > -\infty} \frac{e_i}{\rho_i} \log e_i + \sum_{i:\rho_i = 0} \sum_j b_{ij}^* \log b_{ij}^* - \sum_i e_i.
\end{align*}
Since we know $\mathbf{b}^* = \mathbf{b}(\mathbf{p}(\mathbf{b}^*_{=0}, \mathbf{b}^*_{<0}))$,  this leads to equality in \eqref{eq::dual::4} and \eqref{eq::dual::5} in this case. Therefore,
\begin{align*}
\Upsilon(\mathbf{p}(\mathbf{b}^*_{=0}, \mathbf{b}^*_{<0}))&= \sum_{i : 0 > \rho_i > -\infty}  \Bigg( \frac{1}{\rho_i} \sum_j b_{ij} \log  a_{ij} \frac{b_{ij}^{\rho_i - 1}}{(p_j( \mathbf{b}^*_{=0}, \mathbf{b}^*_{<0}))^{\rho_i}}+ \frac{e_i}{\rho_i} \log e_i \Bigg) \\
&\hspace*{0.6in}+ \sum_{i: \rho_i = -\infty} \sum_j b_{ij}\log  \frac{b_{ij}}{p_j(\mathbf{b}^*_{=0}, \mathbf{b}^*_{<0}) c_{ij}} \\
&\hspace*{0.6in} + \sum_{i: \rho_i = 0} \sum_j b_{ij}^* \log \frac{b^*_{ij}}{p_j(\mathbf{b}^*_{=0}, \mathbf{b}^*_{<0})} + \sum_{j} p_{j}(\mathbf{b}^*_{=0}, \mathbf{b}^*_{<0}) - \sum_i e_i \\
&= -\Phi(\mathbf{b}^*_{=0}, \mathbf{b}^*_{<0}) +  \sum_{j} p_{j}(\mathbf{b}^*_{=0}, \mathbf{b}^*_{<0})  + \sum_{i : \rho_i > -\infty} \frac{e_i}{\rho_i} \log e_i + \sum_{i:\rho_i = 0} \sum_j b_{ij}^* \log b_{ij}^* - \sum_i e_i.
\end{align*}
Since $\sum_{j} p_{j}(\mathbf{b}^*_{=0}, \mathbf{b}_{<0})  = \sum_{j} p_{j}(\mathbf{b}^*_{=0}, \mathbf{b}^*_{<0})  = \sum_i e_i$, the theorem follows.
\end{proof}
\section{Correspondence between Market Equilibrium and Minimal Point, Maximal Point and Saddle Point}
\label{correspondence}
\begin{theorem}
The minimal point in the substitutes domain, the maximal point in the complementary domain, and the saddle point in the mixed case all corresponds to the respective market equilibria.
\end{theorem}
\begin{proof}
Suppose we have a market equilibrium $(\mathbf{b}, \mathbf{p})$. At a market equilibrium, each buyer maximizes her utility function. With  some calculation, we obtain:
\begin{itemize}
\item for $\rho_i = 1$: $b_{ij} > 0$ only if $j$ maximizes $\left\{\frac{a_{ij}}{p_j} \right\}$;
\item for $\rho_i = 0$: $b_{ij} = \lambda_i a_{ij}$; 
\item for $\rho_i = -\infty$: $b_{ij} = \lambda_i c_{ij} p_j$;
\item for other $\rho_i$: $b_{ij} = \lambda_i a_{ij} \left( \frac{b_{ij}}{p_j} \right)^{\rho_i}$.
\end{itemize}

Also, note that for the function $\Phi(\cdot)$, for those $i$ for which $\rho_i \neq \{0, -\infty\}$,
\begin{align*}
\nabla_{b_{ij}} \Phi(\mathbf{b}) &= - \frac{1}{\rho_i} \log a_{ij} - \frac{\rho_i - 1}{\rho_i}(\log b_{ij} + 1) + \log p_j(\mathbf{b}) + \sum_{h}  b_{hj} \frac{1}{p_j(\mathbf{b})} \\
&= \frac{1}{\rho_i}\Big(1 - \log \frac{a_{ij} b_{ij}^{\rho_i - 1}}{[p_j(\mathbf{b})]^{\rho_i}}\Big);
\end{align*}
and for those $i$ for which $\rho_i = -\infty$,
\begin{align*}
\nabla_{b_{ij}} \Phi(\mathbf{b}) = - \log \frac{ b_{ij}}{c_{ij} p_j(\mathbf{b})}.
\end{align*}

Remember that  in both case prices are the sum of the spending. So it's easy to verify that the optimal condition for $\Phi(\cdot)$ will be a market equilibrium, and vice versa. 
The result follows.
\end{proof}
\end{document}